\documentclass[11pt]{article}
\usepackage{fullpage,graphicx,algorithmic,algorithm,enumerate,epsfig,wrapfig}
\usepackage{color,amssymb,amsthm,thmtools}
\usepackage{hyperref}
\usepackage{amsmath}
\usepackage{geometry}
\declaretheorem[parent=section]{theorem}

\declaretheorem[sibling=theorem]{corollary}

\declaretheorem[sibling=theorem]{lemma}

\DeclareMathOperator {\E}{\mathbb{E}}
\let\epsilon=\varepsilon
\newcommand{\poly}{\mathrm{poly}}

\newcommand{\q}{\widehat{d}}
\DeclareMathOperator*{\argmin}{\arg\!\min}

\title{Directory Reconciliation}
\author{Michael Mitzenmacher\thanks{Harvard University School of Engineering and Applied Sciences.  email: michaelm@eecs.harvard.edu.  Michael Mitzenmacher was supported in part by NSF grants CNS-1228598, CCF-1320231, CCF-1563710 and CCF-1535795.} \and Tom Morgan\thanks{Harvard University School of Engineering and Applied Sciences.  email: tdmorgan@seas.harvard.edu.  Tom Morgan was supported in part by NSF grants CNS-1228598 and CCF-1320231.}}
\date{}

\begin{document}

\maketitle

\begin{abstract}
We initiate the theoretical study of \emph{directory reconciliation}, a generalization of document exchange, in which Alice and Bob each have different versions of a set of documents that they wish to synchronize.  This problem is designed to capture the setting of synchronizing different versions of file directories, while allowing for changes of file names and locations without significant expense.  We present protocols for efficiently solving directory reconciliation based on a reduction to document exchange under edit distance with block moves, as well as protocols combining techniques for reconciling sets of sets with document exchange protocols.  Along the way, we develop a new protocol for document exchange under edit distance with block moves inspired by noisy binary search in graphs, which uses only $O(k \log n)$ bits of communication at the expense of $O(k \log n)$ rounds of communication.
\end{abstract}

\section{Introduction}
Document exchange is a well studied two party communication problem, in which Alice and Bob have documents (strings) $a$ and $b$ respectively and they wish to communicate efficiently so that Bob can recover $a$.  They are given a small bound $k$ on the edit distance (or edit distance with block moves) between $a$ and $b$, and wish to use communication proportional to $k$, rather than to the lengths of their strings.  This has immediate applications to version control software, in which a server and client wish to synchronize different versions of the same files.  In such a setting, document exchange would match each file with the corresponding file with the same name on the other side;  this can be done in parallel.  However, if files are allowed to change names or locations in the file structure, the this approach could introduce significant inefficiencies, as a large file whose name is changed may have to be transmitted in its entirety between the parties.

We introduce the problem of \emph{directory reconciliation} to address this issue.  In this problem Alice and Bob each have a \emph{directory}, which we define to be a set of documents.  We have a small bound $d$ on the number of edits (character insertions, deletions, and substitutions) required to transform Alice's directory into Bob's.\footnote{Note that we use $k$ as our edit bound for document exchange and $d$ for our edit bound for directory reconciliation.  This is for historical consistency, and it helps to keep clear which problem we are solving.}  The goal of directory reconciliation is for Bob to recover Alice's directory using as little communication as possible.  In our version control application, a file's name and location could be encoded as a prefix for the document, and now a one character change to a file's name corresponds only to a single edit, rather than the deletion of a whole file and the creation of a new one.  Directory reconciliation is also applicable to situations where we wish to synchronize file collections that are closely related but do not have file names linking them.

We present two approaches to solving directory reconciliation.  The first approach is designed to minimize total communication; we accomplish this by using more rounds of communication.  To achieve this, we show that directory reconciliation can be reduced to document exchange under edit distance with block moves.  Recall that a block move operation selects a contiguous substring of any length, deletes it from its current location and inserts it elsewhere in the string.  The state of the art for this form of document exchange is the IMS sketch \cite{irmak2005improved} which is a one round protocol using $O(k \log n \log(n/k))$ bits of communication.  In \autoref{sec:multiround} we present our main technical result, a new protocol for document exchange under edit distance with block moves that achieves $O(k \log n)$ bits of communication at the expense of using more rounds communication ($O(k \log n)$ of them).  This protocol draws inspiration from techniques for noisy binary search, but requires new techniques and analysis in order to meet our communication requirements.

Our second approach is designed to be more communication efficient in terms of the number of rounds.  This approach is based on combining document exchange protocols with techniques for reconciling sets of sets \cite{mitzenmacher2017reconciling}.  We provide one round directory reconciliation protocols that are generally superior to what we achieve from our reduction to the IMS sketch.  In particular, they perform significantly better when the directory consists of a large number of small files.  Additionally, we provide efficient directory reconciliation protocols that use only a constant number of rounds for the setting when the bound $d$ is unknown.  These results motivate studying directory reconciliation problem as a distinct problem from document exchange, as they improve upon what is possible by direct reduction to document exchange.

\subsection{Related Work}
The formal study of document exchange began with Orlitsky \cite{orlitsky1993interactive} and has received significant attention since then;  see, for example, \cite{belazzougui2016edit,irmak2005improved,chakraborty2016streaming}.  We summarize the current state-of-the-art document exchange protocols in \autoref{sec:docexchange}.  All of the modern protocols use only a single round of communication.  While there was a line of work on multi-round protocols \cite{schwarz1990low, cormode2000communication, langford2001multiround, suel2004improved}, these protocols are dominated by the one round IMS sketch \cite{irmak2005improved}, which incorporates the ideas behind them. The rsync algorithm \cite{rsync, rsyncalg} is a well known practical tool for synchronizing files or directories. However, rsync and related tools have poor worst case performance with regards to the amount of data they communicate.  Furthermore, rsync performs directory synchronization by individually synchronizing files with common names/locations and thus behaves poorly if a file's name/location changes, an issue we seek to remedy with our model of directory reconciliation.

Set reconciliation is a related problem in which Alice and Bob each have a set, and they wish to communicate efficiently so that Bob can recover Alice's set given knowledge of a small bound $d$ on the difference between their sets \cite{minsky2003set,starobinski2003efficient}.  One of the solutions to set reconciliation makes use of the Invertible Bloom Lookup Table (IBLT) \cite{eppstein2011straggler,gm11}.  We make extensive use of IBLTs in several of our protocols, and describe them further in \autoref{sec:iblt}.

Set of sets reconciliation extends set reconciliation to the scenario where Alice and Bob's set elements are themselves sets, as proposed in \cite{mitzenmacher2017reconciling}.  Here the bound $d$ is on the number of element additions and deletions needed to make their sets of sets equal.  Mitzenmacher and Morgan \cite{mitzenmacher2017reconciling} develop several protocols for this problem that we adapt to the setting of document exchange, once again making heavy use of IBLTs.

Noisy binary search is the problem of searching via comparison for an item in a list or graph in which those comparisons have some chance of returning an incorrect response \cite{ben2008bayesian,emamjomeh2016deterministic}.  Solutions for this include multiplicative weights based algorithms which incrementally reinforce the likelihood that each possible candidate is the target.  Our multi-round document exchange protocol draws heavily from these ideas, by using noisy hash based comparisons to incrementally discover common substrings of Alice and Bob's documents.  More specifically, we repeatedly use a variation of noisy binary search to find the longest substring of Bob's document that is a prefix of Alice's document.  While similar to Nisan and Safra's use of noisy binary search to find the longest common prefix of two strings \cite{nisan1993communication}, the expansion to substrings greatly complicates realizing our desired communication bound.

\vspace{-8pt}
\section{Preliminaries}
\label{sec:prelims}
We focus on two problems, directory reconciliation and document exchange.  In the problem of directory reconciliation, Alice and Bob each have a directory, represented as a set of at most $s$ documents, each of which is binary string of length at most $h$.\footnote{We use a binary alphabet here for convenience.  Our results easily extended to larger alphabets.}  Recall that to interpret a file directory as a set, the file's name and directory location are encoded as a prefix of the document.  Hence moving a file or changing its name corresponds to a small number of character edits to the corresponding document in our set. 

The sum of the sizes of each parties' documents is at most $n$.  Bob's directory is equal to Alice's after a series of at most $d$ edits (single character insertions, deletions, and substitutions) to Alice's documents.  Let $\q$ be an upper bound on the number of documents that differ between Alice and Bob.  In general, we may not have such a bound in which case we use $\q = \min(d,s)$.  We develop protocols designed to terminate with Bob fully recovering Alice's directory.

In document exchange, Alice and Bob each have a binary string ($a$ and $b$ respectively) of length at most $n$.  We have a bound $k$ on either $\Delta_e(a,b)$, the edit distance between $a$ and $b$, or $\Delta_{\bar{e}}(a,b)$, the edit distance with block moves.  $\Delta_e(a,b)$ is equal to the minimum number of character insertions, deletions and substitutions to transform $a$ into $b$.  $\Delta_{\bar{e}}(a,b)$ is equal to the minimum number of block moves, character insertions, deletions, and substitutions to transform $a$ into $b$.  The goal of a document exchange protocol is to allow Bob to recover $a$ as efficiently as possible.

Throughout this paper, we work in the word RAM model, with words of size $\Theta(\log n)$ for both problems.  We refer to the number of rounds of communication in a protocol for the total number of messages sent.  A one round protocol therefore consists of a single message from Alice to Bob.

All of our protocols use the public randomness model, meaning that any random bits used in the protocol are shared between Alice and Bob automatically, without additional communication.  This simplifies our presentation as our protocols make heavy use of various hash functions, and public randomness allow Alice and Bob to be able to use the same hash functions without communication anything about them.  Our protocols can be converted to the private randomness model using minimal additional communication via standard techniques \cite{newman1991private}.  In practice, one would instantiate the public randomness model by sharing of a small random seed to be used for generating all of the random bits used in the protocol.

\subsection{Document Exchange Protocols}
\label{sec:docexchange}
Our protocols often use existing document exchange protocols as subroutines.  Here we review the current state-of-the-art document exchange protocols for the setting of edit distance, and edit distance with block moves.\footnote{At the time of this writing, there is a newly released protocol for document exchange under edit distance without block moves of \cite{haeupler2018optimal}.  As this paper is still in pre-print form, and in particular currently lacks a concrete running time for its document exchange protocol, we have opted not to discuss it here.}  The following is the best known protocol for document exchange under edit distance with block moves.

\begin{theorem}[Theorem 1 of \cite{irmak2005improved}] \label{thm:ims}
Document exchange under edit distance with block moves can be solved in one round using $O(k \log n \log (n / k))$ bits of communication and $O(n \log (n / k))$ time, with probability at least $1-1/\poly(n)$.
\end{theorem}

The protocol for this theorem is fairly simple.  For each of $\Theta(\log (n / k))$ levels, Alice transmits an encoding of $a$.  For each level $i$  in $[\Theta(\log (n / k))]$, she splits her string into $2^i k$ blocks, computes a $\Theta(\log n)$ bit hash of each one, and encodes them using the systematic part of a systematic error correcting code for correcting $O(k)$ errors.  At the bottom level, where the blocks are of size $\Theta(\log n)$, she encodes the blocks themselves rather than hashes of them.  After receiving this message, Bob iterates through the levels.  Since the first level has $O(k)$ blocks, he can decode it immediately and recover Alice's hashes.  He applies a rolling hash to every length $n/k$ contiguous substring of $b$ and if any of Alice's hashes match any of his own, then by inverting the hash (assuming no hash collisions) he knows the contents of that block of $a$.  Bob then continues through the levels, each time using what he has recovered so far of Alice's string to decode the next level's code.  Since there are $O(k)$ edits between their strings, it can be shown that each level will only have $O(k)$ hashes that aren't present in Bob's string.  By decoding the final level, Bob will have recovered Alice's whole string.  The protocol only fails if there are hash collisions, which, given the hash size and number of strings hashed, occurs with probability $1-1/\poly(n)$.

Now we turn to the best known protocol for document exchange under edit distance without block moves.

\begin{theorem}[Theorem 9 of \cite{belazzougui2016edit}] \label{thm:de_best}
Assuming $k < n^{\epsilon}$ for a sufficiently small constant $\epsilon > 0$, document exchange under edit distance can be solved in one round using $O(k(\log^2 k + \log n))$ bits of communication and $O(n (\log k+\log \log n))$ time, with probability at least $1 - 1 / \poly(k \log n)$.
\end{theorem}

Note that for the case where $k \geq n^{\epsilon}$, \autoref{thm:ims} represents the best known document exchange protocol, even under edit distance without block moves.  \autoref{thm:de_best} is based on a careful application of the CGK encoding of \cite{chakraborty2016streaming}, which embeds from the edit distance space into Hamming space.  \autoref{thm:de_best} results from applying this encoding $O(\log \log n)$ times, each time matching up common pieces of $a$ and $b$ using the encoding, until at the final level the unmatched strings are of small enough size that applying \autoref{thm:ims} yields the desired bound.

\subsection{Invertible Bloom Lookup Tables}
\label{sec:iblt}
Several of our protocols make use of the Invertible Bloom Lookup Table (IBLT) \cite{gm11}, a data structure representing a set that was designed to solve the set reconciliation problem.  We summarize the structure and its properties here;  more details can be found in \cite{eppstein2011straggler,gm11}.  An IBLT is a hash table with $q$ hash functions and $m$ cells, which stores sets of key-value pairs.  (It can be used to just store keys also.)  We add a key-value pair (each from a universe of size $O(u)$) to the table by updating each of the $q$ cells that the key hashes to.  (We assume these cells are distinct;  for example, one can use a partitioned hash table, with each hash function having $m/q$ cells.)  Each cell has a number of entries: a count of the number of keys hashed to it, an XOR of all of the keys hashed to it, an XOR of a checksum of all of the keys hashed to it, and, if we are using values, an XOR of all the values hashed to it.  The checksum, produced by another hash function, is $O(\log u)$ bits and is meant to guarantee that with high probability, no cells containing distinct keys will have a colliding checksum.  We can also delete a key-value pair from an IBLT through the same operation as adding it, except that now the counts are decremented instead of incremented.

An IBLT is \emph{invertible} because if $m$ is large enough compared to $n$, the number of key-value pairs inserted into it, we can recover those $n$ pairs via a peeling process.  Whenever a cell in the table has a count of 1, the key XOR will be equal to the unique key inserted there, and the value XOR will be equal to the unique value inserted there.  We can then delete the pair from the table, potentially revealing new cells with counts of 1 allowing the process to continue until no key-value pairs remain in the table.  This yields the following theorem. 

\begin{theorem}[Theorem 1 of \cite{gm11}] \label{thm:iblt}
There exists a constant $c > 0$ so that an IBLT with $m$ cells ($O(m \log u)$ space) and at most $cm$ key-value pairs will successfully extract all keys with probability at least $1 - O(1/\poly(m))$.
\end{theorem}

A useful property of IBLTs is that we can ``delete'' pairs that aren't actually in the table, by allowing the cells' counts to become negative.  In this case, the IBLT represents two disjoint sets, one for the inserted or ``positive'' pairs and one for the deleted or ``negative'' pairs.
This addition requires a minor modification to the peeling process, which allows us to extract both sets.  Now just as we peeled cells with 1 counts by deleting their pairs from the table, we also peel cells with $-1$ counts by adding theirs pairs to the table.  Unfortunately, a cell with count of 1 or $-1$ might have multiple pairs (some from each set) hashed there, whose counts only add up to $\pm 1$. However, we remedy this issue by using our checksums.  With high probability, a cell with a count of $\pm 1$ will actually represent only a single pair if and only if the checksum of the cell's key XOR equals the cell's checksum XOR.

This property of IBLTs allows us to insert each of the items in a single set into the IBLT, then delete the items from another set from the IBLT.  Inverting the IBLT then reveals the contents of the symmetric set difference of the original two sets, and by \autoref{thm:iblt} will succeed with high probability so long as the size of this set difference is at most $cm$.

In most of our uses of the IBLT, we only have keys, and no associated values.  As such, unless otherwise noted, assume that our IBLTs lack value fields, and when we insert or delete an item from the IBLT, we are treating it as a key.  We sometimes refer to ``encoding'' a set in an IBLT as inserting all of its elements into it.  We similarly ``decode'' a set difference from an IBLT by extracting its keys.  

There is one final nice property of IBLTs that we exploit.  Let $T_1$ and $T_2$ be two IBLTs with the same number of cells and the same hash functions.  Let $S_1$ and $S_2$ be two sets, and we insert the items of $S_1$ into $T_1$ and the items of $S_2$ into $T_2$.  If $S_1$ and $S_2$ are disjoint, then we can ``add'' $T_1$ and $T_2$ together to make a single IBLT encoding $S_1 \cup S_2$.  We do this iterating through $i \in [m]$, adding the $i$th cells of $T_1$ and $T_2$ together by summing their counts and XORing their other fields.  Similarly, we can ``subtract'' $T_1$ and $T_2$ to yield a single IBLT encoding the symmetric set difference of $S_1$ and $S_2$.

\subsection{Notation}

Given a (1-indexed) vector $s$ of length $n$, we will use the notation $s_{i:j}$ to refer to subset of $s$ consisting of indices $i$ through $j$ inclusive.  Similarly, $s_{:i}$ refers to the length $i$ prefix of $s$ and $s_{i:}$ refers to the the length $n-i+1$ suffix of $s$.

We will frequently use $\Theta(\log n)$ bit hashes as identifiers for strings.  We use the property that for a sufficiently large constant in the order notation, we have no collisions among at most $\poly(n)$ such hashes with probability at least $1-1/\poly(n)$.

\section{Multi-Round Document Exchange Protocol}
\label{sec:multiround}
Before we provide our own protocol for document exchange, we show that directory reconciliation can be solved via a straightforward reduction to document exchange under edit distance with block moves.  The current state-of-the-art protocol for this problem is \autoref{thm:ims}, and this reduction provides a baseline against which we compare the rest of our protocols.

\begin{restatable}{theorem}{thmreduction} \label{thm:reduction}
Directory reconciliation can be solved in one round using $O(d \log n \log (n/d))$ bits of communication and $O(n \log(n / d))$ time with probability at least $1-1/\poly(n)$.
\end{restatable}

The protocol for this starts with each party computing a $\Theta(\log n)$ bit hash of each of their documents.  They concatenate their documents together, with each pair separated by a random $\Theta(\log n)$ bit delineation string, then perform document exchange on their concatenated documents via \autoref{thm:ims}, and finally Bob decomposes Alice's concatenated document into Alice's directory.  As we argue in \autoref{app:reduction}, the edit distance with block moves between the concatenated documents is at most $2d$, and thus \autoref{thm:ims} yields the desired bounds.

Now we develop a protocol for document exchange under edit distance with block moves that achieves a communication cost of $O(k \log n)$.  We do this with a multi-round protocol, inspired by noisy binary search algorithms \cite{emamjomeh2016deterministic}, that identifies the common blocks between Alice and Bob's strings via many rounds of back and forth communication.

\begin{theorem} \label{thm:multiround}
Document exchange under edit distance with block moves can be solved in $O(k \log n)$ rounds using $O(k \log n)$ bits of communication and $O(n^2 \log n)$ time, with probability at least $1 - 2^{-\Theta(k\sqrt{\log n})}-1/\poly(n)$.
\end{theorem}

This implies a protocol for directory reconciliation via the same reduction as in \autoref{thm:reduction}.

\begin{corollary} \label{cor:multiround}
Directory reconciliation can be solved in $O(d \log n)$ rounds using $O(d \log n)$ bits of communication and $O(n^2 \log n)$ time, with probability at least $1 - 2^{-\Theta(d\sqrt{\log n})}-1/\poly(n)$.
\end{corollary}

The main technical work in our protocol comes from the following lemma.

\begin{lemma} \label{lem:prefixprotocol}
Given $\kappa \leq n$, Alice can find the longest prefix of her document up to length $n / \kappa$ that is a contiguous substring of Bob's document in $O(\log n)$ rounds using $O(\log n)$ bits of communication and $O((n^2 / \kappa) \log n)$ time, with probability at least $1 - 2^{-\sqrt{\ln n}}$.
\end{lemma}

We prove \autoref{lem:prefixprotocol} later, building off of techniques from noisy binary search in graphs \cite{emamjomeh2016deterministic}.  We use this lemma as a subroutine in our protocol for \autoref{thm:multiround} by incrementally building up a larger and larger prefix of Alice's document that is known to Bob.  Along the way \autoref{lem:prefixprotocol} may fail due to its own internal randomness, but we make no assumptions on what mode that failure takes.  \autoref{lem:prefixprotocol}'s protocol may abort and report failure, it may report a prefix of Alice's document that is too long, and thus is not a substring of Bob's document, or it may report a prefix that is a substring of Bob's document but is not the longest possible one.  We show that so long as most of our applications of \autoref{lem:prefixprotocol} succeed, no matter what form the failures take, our resulting protocol for document exchange will succeed.

\begin{proof}[Proof of \autoref{thm:multiround}]
Let $a \in \{0,1\}^n$ be Alice's document and $b \in \{0,1\}^n$ be Bob's document.

Our protocol will happen in $t\leq 10k$ phases.  After each phase, Bob will have recovered a progressively larger prefix of $a$.  After the $i$th phase he will have recovered $a^{(i)}$, where $a^{(i)}$ is a prefix of $a$ and $a^{(t)} = a$ with probability at least $1 - 2^{-\Theta(k\sqrt{\log n})}-1/\poly(n)$.  We use $a - a^{(i)}$ to refer to the string $a$ after removing the prefix $a^{(i)}$.  Given a string $s$, we use $|s|$ to refer to the length of $s$.

In the $i$th phase, Alice and Bob uses \autoref{lem:prefixprotocol} (setting $\kappa = k$) to find the largest prefix $s$ (up to length $n / k$) of $a - a^{(i-1)}$ contained in $b$.  If $s$ is shorter than $\log_2 n$, Alice directly transmits the first $\log_2 n$ bits of $a - a^{(i-1)}$, thus $|a^{(i)}| = |a^{(i-1)}| + \log_2 n$.  Otherwise, Alice transmits $|s|$, along with a $\Theta(\log n)$ bit hash of $s$ to Bob, who compares it to to the hash of each length $|s|$ substring in $b$.  With probability at least $1 - 1/\poly(n)$, the hash of $s$ matches only one unique substring of $b$, and that is $s$, thus Bob has recovered $a^{(i)} = a^{(i-1)} + s$.  If at any point a failure occurs, either one detected in \autoref{lem:prefixprotocol} or if the hash of $s$ does not have a unique match, we move on letting $a^{(i)} = a^{(i-1)}$.

First we argue that assuming no hashing failures or failures in \autoref{lem:prefixprotocol}, the protocol succeeds with $t \leq 5k$.  The argument follows that of Lemma 3.1 of \cite{irmak2005improved}.  We imagine that $b$ is written on a long piece of paper, and we perform each edit operation to transform $b$ into $a$ by cutting the paper, rearranging pieces and inserting individual characters for insert and substitution operations.  Each operation requires at most 3 cuts, so $a$ consists of the concatenation of at most $3k$ substrings of $b$ plus up to $k$ newly inserted characters.  In each phase of the algorithm, we either recover up to the end of one of the substrings, at least one of the inserted characters, or $n / k$ characters.  The first case can happen at most $3k$ times, the second and third cases can each happen at most $k$ times, so at most $5k$ phases are required to recover $a$ in its entirety.

There are two ways that a phase can fail; either the hash of $s$ can match with a substring of $b$ that does not equal $s$, or \autoref{lem:prefixprotocol} can fail.  Union bounding over the $t \leq 10k$ phases, with probability at least $1 - 1/\poly(n)$, none of the hashes mismatch.  \autoref{lem:prefixprotocol} can fail one of three ways: it can report an $s$ which is too short, report an $s$ that is too long, or it can fail to find an $s$ at all.  If it fails to find an $s$, we have case where $a^{(i)} = a^{(i-1)}$, basically meaning we try again with fresh randomness.  If $s$ is too long, then assuming no hash mismatches, Bob will not find a substring of $b$ that matches $s$'s hash so we once again have $a^{(i)} = a^{(i-1)}$.  Finally if $s$ is too short, $|a^{(i)}| > |a^{(i-1)}|$ but we won't have completed one of the at most $3k$ substrings of $b$ of which $a$ is comprised.

So long as \autoref{lem:prefixprotocol} succeeds $5k$ times in at most $10k$ phases, the protocol succeeds.  Since each application of \autoref{lem:prefixprotocol} succeeds with probability at least $1 - 2^{\sqrt{\ln n}}$, by a Chernoff bound, we succeed at least $5k$ times in $10k$ attempts with probability at least $1 - 2^{-k(\sqrt{\ln n}-2)}$ so our overall success probability is at least $1 - 2^{-\Theta(k\sqrt{\log n})}-1/\poly(n)$ as desired.

Each phase uses $O(\log n)$ rounds totaling $O(\log n)$ bits of communication and takes $O((n^2 / k) \log n)$ time, thus the whole protocol takes $O(k \log n)$ rounds, uses $O(k \log n)$ bits of communication and takes $O(n^2 \log n)$ time. 
\end{proof}

Now we sketch the protocol for \autoref{lem:prefixprotocol}, with the details presented in Algorithm \ref{alg:prefixprotocol}.  In our protocol, Bob will represent all contiguous substrings of his document of length at most $n/\kappa$ using a tree $T$.  Each node in $T$ represents a substring of $b$.  The root of $T$ corresponds to the empty string, and each node at depth $r$ corresponds to a unique substring of length $r$.  A node at depth $r$'s parent is the node corresponds to its length $r-1$ prefix.  $T$ is essentially an uncompressed suffix tree \cite{gusfield1997algorithms} formed from the reverse of each of the length $n/\kappa$ substrings of $b$.

\begin{algorithm}[h!]
\caption{$\textsc{MultWeightsProtocol}(a,I,t,T)$: Alice inputs $(a, I, t)$. Bob inputs $(T, I, t)$.}
\label{alg:mwp}
\begin{itemize}
\item Alice and Bob: $t \gets 1$ and $q_1 \gets 1$.  
\item Bob: $w(u) \gets 1$ for all $i \in t$.  $M \gets \emptyset$.
\item For $i = 1$ to $t$:
\begin{itemize}
\item Bob: if $\exists u \in T$ such that $w(u) > w(T)/2$: 
\begin{itemize}
\item Bob: $w(u) \gets 0$, $M \gets M \cup \{u\}$ and $\eta \gets 0$.
\item Bob: send $\eta$ to Alice.
\end{itemize}
\item Bob: otherwise:
\begin{itemize}
\item Bob: $\rho \gets \argmin_r \max\left(w(T_{-r}), \max_{u \in D_r} w(T_u)\right)$. Sets $o$ such that $\rho \in I_o$.
\item Bob: if $o \geq q_1$, $q_1' \gets \min(o, q_1+3)$.  Otherwise, $q_1' \gets \max(o, q_1-3)$.
\item Bob: $\ell_1' \gets I_{q_1',\left\lceil|I_{q_1'}|/2\right\rceil}$ (the midpoint of $I_{q_1'}$) and $\eta \gets 1$.
\item Bob: if $q_1' = o$ and $\ell_1'$ is not balanced with respect to $w$, $\eta \gets 2$, we set either $q_2' \gets q_1' +1$ or $q_2' -1$ so that $\ell_1'$ and $\ell_2' = I_{q_2',\left\lceil|I_{q_2'}|/2\right\rceil}$ straddle $\rho$, meaning that $\min(\ell_1',\ell_2') < \rho < \max(\ell_1',\ell_2')$.
\item Bob: send $\eta$ and $q_j' - q_j$ for $j \in [\eta]$ to Alice.
\end{itemize}
\item Alice and Bob: for $j \in [\eta]$: $q_j \gets q_j'$, $m_j \gets \left\lceil|I_{q_j}|/2\right\rceil$, $\ell_j \gets I_{q_j,m_j}$.
\item Alice: for $j \in [\eta]$, send $h(a_{:\ell_j})$ to Bob.
\item For $j \in [\eta]$: for $u \in D_{\ell_j}$:
\begin{itemize}
\item Bob: if $h(a_{:{\ell_j}}) = h(u)$, for all $v \in T_u$, $w(v) \gets w(v) \cdot p$ and for all $v \in T - T_u$, $w(v) \gets w(v) \cdot (1-p)$.
\item Bob: if $h(a_{:{\ell_j}}) \neq h(u)$, for all $v \in T_u$, $w(v) \gets w(v) \cdot (1-p)^2$ and for all $v \in T - T_u$, $w(v) \gets w(v) \cdot p^2$.
\end{itemize}
\item Alice and Bob: split $I$ around the queried points.  Specifically, for $j \in [\eta]$ in descending order of $q_j$, if $|I_{q_j}| > 1$ then $I \gets I_{:q_j-1} + I_{q_j,:m_j-1} + I_{q_j,m_j:} + I_{q_j+1:}$
\end{itemize}
\item Bob: send to Alice a $2t$-bit vector $z$ indicating for each phase whether any nodes at depth $\ell_1$ or $\ell_2$ are in $M$.
\item Alice and Bob: construct $I'$ such that $|I'|=1$ and $I'_1$ is the ordered list of unique depths of nodes in $M$.
\item Alice: return $I'$.
\item Bob: return $I'$ and $M$.
\end{itemize}
\end{algorithm}

\begin{algorithm}[h!]
\caption{Alice inputs $a$, $n$ and $\kappa$. Bob inputs $b$, $n$ and $\kappa$.}
\label{alg:prefixprotocol}
\label{alg:}
\begin{itemize}
\item Bob: construct $T_1$, the tree representing all contiguous substrings of $b$ up to length $n/\kappa$.
\item Alice and Bob: $\delta \gets 2^{-\sqrt{\ln n}} / 3$ and $t_1 \gets 85 \ln n + 63 \ln(1/\delta)$.
\item $I_1,M_1 \gets \textsc{MultWeightsProtocol}(a,[n],t_1,T_1)$
\item Bob: construct $T_2$, the tree representing the hierarchy of the nodes in $M_1$.
\item Alice and Bob: $t_2 \gets 82 \ln t_1 + 63 \ln(1/\delta)$.
\item $I_2,M_2 \gets \textsc{MultWeightsProtocol}(a,I_1,t_2,T_2)$
\item For $r \in I_2$ (in decreasing order):
\begin{itemize}
\item Alice: send $\Theta(\log t_2 + \log (1/\delta))$-bit hash $h'(a_{:r})$ to Bob
\item Bob: if $\exists u \in M_2 \cap D_r$ such that $h'(u) = h'(a_{:r})$, send $1$ to Alice and return $u$
\item Alice: if received $1$, return $r$
\end{itemize}
\item Alice and Bob: return failure
\end{itemize}
\end{algorithm}

We use $D_r$ to refer to the set of nodes in $T$ at depth $r$.  Given a node $u$, $T_u$ refers to the subtree rooted at $u$ in $T$.  We use $T_{-r}$ to refer to the tree above depth $r$.  In other words $T_{-r} = T - \cup_{u\in D_r} T_u$.

Our protocol will consist of $\Theta(\log n)$ phases.  In the $i$th phase, we draw an $O(1)$ bit hash function $h$, and Alice sends $h(a_{:\ell})$ to Bob.  Bob then compares this hash to the depth $\ell$ nodes in $T$ (the length $\ell$ substrings of $b$).  Bob then responds with information to determine what $\ell$ to use in the next phase.  After $\Theta(\log n)$ rounds of this, we hope to gain enough information from these hashes for Bob to confidently determine the length of the largest prefix of $a$ which is contained in $b$.

The set of prefixes of $a$ which are contained in $b$ correspond to a path from the root in $T$.  We are searching for a single node $g$ in $T$, which is the last node in this path.  Whenever $h(a_{:\ell})$ does not match the hash of a node $u$ in $T$, we know that $g$ is not in the subtree rooted at $u$.  When the hash does match, then $g$ is more likely to be in the subtree rooted at $u$, and less likely to be in the rest of $T$.  This is essentially the kind of feedback used in noisy binary search in graphs \cite{emamjomeh2016deterministic}.  If we were unconstrained in the amount of communication Bob sends to Alice to determine the next $\ell$, and Bob could simply send $\Theta(\log n)$ bits to exactly specify the optimum $\ell$, then we could reduce directly to the algorithm of \cite{emamjomeh2016deterministic}.  However, in order to achieve the lemma we must restrict ourselves to $O(1)$ bits of communication per phase.  We emphasize that this is where most of our technical contribution lies, as it requires new techniques and analysis to handle this restriction.

Both Alice and Bob maintain a partitioning of the range $[n]$ into intervals $I_1,\ldots,I_{t}$.  Before the first phase, they just have one interval containing all of $[n]$.  In each phase, Alice and Bob agree to query in the $q$th interval.  Alice then chooses $\ell$ to be the midpoint of $I_{q}$.  After transmitting the $h(a_{:\ell})$, the interval containing $\ell$ will be split in two around $\ell$.  Bob then tells Alice which interval to choose the next $\ell$ from, by transmitting $O(1)$ bits specifying how much to add or subtract from $q$.

We choose the size of $h_i$ so that the probability of two unequal strings having matching hash values is at most $1-p = 1/3$.  For each node $u \in T$, Bob maintains a weight $w(u)$.  All the weights are initialized to $1$, and Bob updates them in response to the hashes he receives from Alice.  For each $u \in D_{\ell}$, if $h(a_{:\ell}) = h(u)$, we multiply each of the weights in $T_u$ by $p$ and we multiply the weights of all other nodes in $T$ by $1-p$.  If $h(a_{:\ell}) \neq h(u)$, we multiply each of the weights in $T_u$ by $(1-p)^2$ and we multiply the weights of all other nodes in $T$ by $p^2$.  Here we are using the one-sidedness of comparisons, in that if the hashes do not match then we know that the strings do not match with probability 1.  As a result, we could instead zero the weights of the nodes in $T_u$ when $h(a_{:\ell}) \neq h(u)$, but doing so would not improve our asymptotic result and would complicate the analysis.

Whenever a node's weight is at least half of the total tree weight, we zero its weight and add the node to a list $M$.  We show that after $\Theta(\log n)$ rounds, our goal $g$ is likely to be in $M$.  We then repeat this protocol with a new tree $T_2$ consisting only of nodes in $M$.  $T_2$ is formed by taking the transitive closure of $T$, then removing all nodes not in $M$, then taking the transitive reduction of the result.  In other words, in $T_2$, $u$ is a parent of $v$ if $u$ was $v$'s most recent ancestor in $M$.  Note that unlike $T$, $T_2$ need not be a binary tree.  When running the protocol on $T_2$, we now have that $I$ no longer consists of intervals, but is instead contiguous subsets of the node depths present in $M$, and when we select the midpoint of an element of $I$, we choose the median.  We show that by running the  protocol on $T_2$ for $\Theta(\log \log n)$ rounds, $g$ is likely to be in the new candidate list of nodes $M_2$.  $|M_2| = \Theta(\log \log n)$, so at this point, we directly determine which node in $M_2$ is $g$ by sending a $\Theta(\sqrt{\log n})$ bit hash for each node in $M_2$.

Now we argue that \autoref{alg:prefixprotocol} satisfies \autoref{lem:prefixprotocol}. First we note that all of the steps in \textsc{MultWeightsProtocol} are achievable, with the only non-trivial part being provided by the following lemma.

\begin{lemma} \label{lem:mwpcorrect}
In \textsc{MultWeightsProtocol}, Alice can construct $I'$ using $z$.
\end{lemma}

\begin{proof}
Consider a node $u \in D_r$ and its parent $v$.  If $r$ has not been queried in any phase, then $w(u) = w(v)$.  This is because any weight update performed by a query at a depth $r' \neq r$ will update $w(u)$ and $w(v)$ identically.  If $r' > r$ then either $w(u)$ and $w(v)$ will be unchanged or they will both be multiplied by $1-p$.  If $r < r$ then either $w(u)$ and $w(v)$ will both be multiplied by $p$ or they will both be set to $0$.

In order for $u$ to be added to $M$, $w(u) > w(T)/2$.  This is impossible if $w(u) = w(v)$, therefore $r$ must have been queried in at least one phase.  Thus, every depth $r$ that could be in $I'$ will be included in $z$ so Alice can construct $I'$.
\end{proof}

We say that a depth $r$ is \emph{queried} in a phase $i$ of Algorithm \ref{alg:mwp} if $\ell_1 = r$ or $\ell_2 = r$ in that phase.  A query of depth $r$ partitions the the tree into $|D_r|+1$ components: $T_{-r}$ along with $T_u$ for each $u \in D_r$.  A query of depth $r$ is \emph{balanced} if each component has at most $3/4$ of the total weight in the tree.  In other words, $\max\left(w(T_{-r}),\max_{u \in D_r} w(T_u)\right) \leq (3/4)w(T)$.  A query of depth $r$ is \emph{informative} if the component containing the target node $g$ has total as most $3/4$ of the total mass.  Specifically, if $g \in T_{-r}$ then $w(T_{-r}) \leq (3/4)w(T)$ and if $g \in T_u$ for $u \in D_r$, then $w(T_u) \leq (3/4)w(T)$.  A balanced query is guaranteed to be informative, but an informative query need not be balanced.  We call phase $i$ of Algorithm \ref{alg:mwp} informative if either an informative query is performed or if a node is added to $M$.

\begin{lemma} \label{lem:balanced}
In \textsc{MultWeightsProtocol}, $\rho = \argmin_r \max\left(w(T_{-r}), \max_{u \in D_r} w(T_u)\right)$ is guaranteed to be a balanced query.
\end{lemma}

\begin{proof}
If $\rho$ is not balanced, then either $w(T_{-\rho}) > (3/4)w(T)$ or there exists a $u \in D_\rho$ such that $w(T_u) > (3/4)w(T)$.  Additionally, by the definition of $\rho$, if it is not balanced then no query is.

Suppose $w(T_{-\rho}) > (3/4)w(T)$.  Let $\rho'$ be the smallest (least deep) query such that $w(T_{\rho'}) = w(T_\rho)$.  By the definition of $\rho$, there must exist a $u' \in D_{\rho'-1}$ for which $w(T_{u'}) \geq w(T_{-\rho}) > (3/4)w(T)$.  $T_{u'}$ consists of $u'$ together with a subset of $\cup_{v \in D_{\rho'}} T_v$.  Since $w(T_{-\rho'}) > (3/4)w(T)$, $\sum_{v \in D_{\rho'}} w(T_v) < (1/4)w(T)$.  Furthermore, every node in the tree has weight at most $w(T)/2$, or it would have been zeroed in a previous step.  Therefore, $w(T_{u'}) \leq w(T)/2 + (1/4)w(T) = (3/4)w(T)$, which is a contradiction so we cannot have that $w(T_{-\rho}) > (3/4)w(T)$.

The argument showing that there cannot exist a $u \in D_\rho$ such that $w(T_u) > (3/4)w(T)$ is essentially identical.
\end{proof}

\begin{lemma} \label{lem:informative}
If $g \in T$, $\textsc{MultWeightsProtocol}(a,I,t,T)$ performs at least $(2t - 8\log_{3/2} \tau)/7$ informative phases, where $\tau$ is the height of $T$.
\end{lemma}
\begin{proof}
Let $r^*$ be the depth of $g$, and $q^*$ be the interval containing $r^*$.  We introduce the potential function $\Phi = |q_1 - q^*| + \alpha (\log_2 |I_q^*| + \log_2 |I_{q^*+1}|)$, where $\alpha = 4 / \log_2(3/2)$.  If $q^*$ is the final interval, then instead $\Phi = |q_1 - q^*| + \alpha (\log_2 |I_q^*| + \log_2 \tau)$.  We argue that in each phase, either a node is added to $M$, an informative query is performed and $\Phi$ increases by at most 5, or $\Phi$ decreases by at least 2.

No phase can ever increase $\Phi$ by more than 5 since $|q'_1 - q_1| \leq 3$ and at most two intervals are split in two, and $|I_q^*|$ and $|I_{q^*+1}|$ are non-increasing.  

If $q'_1 \neq o$ then we have two cases:
\begin{enumerate}
\item $\ell_1'$ is in between $\rho$ and $r^*$.  In this case $\ell_1'$ is informative because $\rho$ is balanced (by \autoref{lem:balanced}) and $\ell_1'$ is closer to $r^*$ than $\rho$ so the component containing $g$ formed by partitioning around $\ell_1'$ must have at most $3/4$ of the total weight. 
\item $\rho$ is in between $\ell_1'$ and $r^*$ or $r^*$ is in between $\ell_1'$ and $\rho$.  Since $q_1' \neq o$, $q_1'$ must be 3 steps closer to $o$ than $q_1$, which in either ordering means it must also be 3 steps closer to $q^*$.  Thus $|q_1 - q^*|$ must decrease by at least 2 (not 3 because splitting an interval can increase the distance by 1).
\end{enumerate}

If $q'_1 = o$, then we have a few cases:
\begin{enumerate}
\item $\ell_1'$ is balanced.  In this case our query is guaranteed to be informative.
\item $r^* < \min(\ell_1',\ell_2')$.  $\min(\ell_1',\ell_2')$ is in between $r^*$ and $\rho$ so it is guaranteed to be informative.
\item $r^* > \max(\ell_1',\ell_2')$.  $\max(\ell_1',\ell_2')$ is in between $r^*$ and $\rho$ so it is guaranteed to be informative.
\item $|I_q^*| \geq 2$ and $\min(\ell_1',\ell_2') \leq r^* \leq \max(\ell_1',\ell_2')$.  Either $q_1' = q^*$ or $q_2' = q^*$ so $|I_q^*|$ will shrink by a factor of at least $2/3$.  $|q_1 - q^*|$ afterwards will be at most $2$, thus in total $\Phi$ will decrease by at least $\alpha \log_2(3/2) - 2 = 2$.
\item $|I_q^*| = 1$ and $\min(\ell_1',\ell_2')=r^*$.  Either $\max(q_1',q_2') = q^*+1$ so $|I_{q^*+1}|$ will shrink by a factor of at least $2/3$.  $|q_1 - q^*|$ afterwards will be at most $2$, thus in total $\Phi$ will decrease by at least $\alpha \log_2(3/2) - 2 = 2$.
\item $|I_q^*| = 1$ and $\max(\ell_1',\ell_2')=r^*$.  $\max(\ell_1',\ell_2')$ is guaranteed to be an informative query since $\rho < \max(\ell_1',\ell_2')$ so $w(T_g) < (3/4) w(T)$.
\end{enumerate}

Initially, $\Phi = 2 \alpha \log_2 \tau$ since there is only one interval.  $\Phi$ is also always non-negative.  Thus, over $t$ phases, if $\iota$ is our number of informative phases then
\begin{equation*}
2 \alpha \log_2 \tau + 5 \iota - 2 (t - \iota) \geq 0,
\end{equation*}
and so $\iota \geq (2 t - 2 \alpha \log_2 \tau) / 7 = (2t - 8\log_{3/2} \tau)/7$.
\end{proof}

Let $X_i$ be $w(g)/w(T)$ at the beginning of phase $i$ in \textsc{MultWeightsProtocol}.  Let $Y_i = \ln(X_{i+1} / X_i)$. We will use these random variables to bound the probability that at the end of \textsc{MultWeightsProtocol}, $w(g) < w(T)/2$ and thus $g$ is never added to $M$.  First we provide some bounds on them, whose proofs appear in \autoref{app:multiproofs}.

\begin{restatable}{lemma}{lemybounded} \label{lem:ybounded}
If no node is added to $M$ in phase $i$, then $-2\ln\left(\frac{p}{1-p}\right) \leq Y_i \leq 6\ln\left(\frac{p}{1-p}\right)$.  If a node other than $g$ is added to $M$ in phase $i$, then $Y_i \geq \ln 2$.
\end{restatable}

\begin{restatable}{lemma}{lemymean} \label{lem:ymean}
If $p \geq 2/3$, and $g$ is not in $M$ by the end of phase $i$, then $E[Y_i] \geq 0$.  Additionally, if phase $i$ is informative, then $E[Y_i] \geq \ln(4p/(2p^2+p+1))$.
\end{restatable}

Using these bounds on the $Y_i$s, together with \autoref{lem:informative} we can apply a Chernoff bound to get the following lemma (whose proof also appears in \autoref{app:multiproofs}).

\begin{restatable}{lemma}{lemmwpprob} \label{lem:mwpprob}
If $p=2/3$, $g \in T$, and $t \geq 79 \ln \tau + 3 \ln |T| + 63 \ln(1/\delta)$, where $\tau$ is the height of $T$, then $\textsc{MultWeightsProtocol}(a,I,t,T)$ will output an $M$ which includes $g$ with probability at least $1-\delta$.
\end{restatable}

Now we put all of these pieces together.

\begin{proof}[Proof of \autoref{lem:prefixprotocol}]
We prove that Algorithm \ref{alg:prefixprotocol} achieves the lemma.  First we argue correctness.

The only step of \textsc{MultWeightsProtocol} where it is not immediate that Alice or Bob has the required information to perform the protocol is whether Bob can construct $I'_1$, which we proved in \autoref{lem:mwpcorrect}.  By \autoref{lem:mwpprob} and our choice of $t_1$ (since $|T_1| \leq n^2$), $M_1$ will contain $g$ with probability at least $1-\delta$.  Since $|M_1| \leq t_1$, together with \autoref{lem:mwpprob}, our choice of $t_2$, and a union bound we have that $M_2$ will contain $g$ with probability at least $1-2\delta$.

With our choice of hash size for $h'$, we can make the probability of hash collisions at most $\delta / t_2 \leq \delta / |M_2|$.  Thus, by a union bound, with probability at least $1-3\delta$ there are no hash collisions and the Alice returns $g$ and Bob returns $|g|$.  Thus by our choice of $\delta$, the protocol succeeds with probability at least $1-2^{-\sqrt{\ln n}}$.

Each phase of \textsc{MultWeightsProtocol} uses $O(1)$ rounds of communication and transmits $O(1)$ bits.  At the end of \textsc{MultWeightsProtocol} there is a single message of size $2t$, thus \textsc{MultWeightsProtocol} performs $O(t)$ rounds of communication and $O(t)$ bits of communication.  Thus, between our two applications of \textsc{MultWeightsProtocol} and the hashes transmitted at the end our rounds of communication are bounded by
\[O(t_1 + t_2 + |I_2|) = O(t_1 + t_2) = O(\log n + \log (1/\delta)) = O(\log n).\]
Our bits of communication are bounded by
\[O(t_1 + t_2 (\log t_2 + \log(1/\delta))) = O(\log n + \log^2(1/\delta)) = O(\log n).\]

The computation bottleneck is the first application of \textsc{MultWeightsProtocol}.  Each phase takes $O(n)$ time compute all of the hash evaluations (using a rolling hash function) and $O(|T_1|)$ time to update all the weights and select the next interval to query, so the total computation time is $O((n+|T_1|) t_1) = O((n^2 / \kappa) \log n)$.
\end{proof}

\section{Sets of Sets Based Protocols} \label{sec:sos_based}

\autoref{cor:multiround} is aimed at minimizing the number of bits of communication, but it does so at the expense of both rounds of communication and computation time.  In this section we develop a single round protocol that is much faster and still outperforms \autoref{thm:reduction} in communication cost for most cases.

\begin{restatable}{theorem}{thmbestsos} \label{thm:best_sos}
Directory reconciliation can be solved in one round using 
\[O(d\log s + d \log^3 \min(d,h) + d \log h \log \min(d,h))\]
bits of communication and 
\[O(nd \log \log h + n \q \log^2 \min(d, h) + n \q \log \min(d,h) \log \log h)\] 
time with probability at least $2/3$.
\end{restatable}

The protocol here is an adaptation of Theorem 3.7 of \cite{mitzenmacher2017reconciling}, using document exchange as a subroutine instead of set reconciliation.  The basic idea is that we represent each document as a pair consisting of a hash of that document, together with the message that \autoref{thm:de_best} (the best known one round document exchange protocol) would send to allow the other party to recover the document.  For each of $\Theta(\log d)$ different levels, we make a different version of this representation for each document, varying the bound $k$ on the edit distance used in the document exchange protocol.  In the $i$th level, starting at $i=1$, $k=2^i$.  Then, we encode all of the document representations in level $i$ into an IBLT of size $\Theta(d/2^i)$.  These IBLTs are the message from Alice to Bob.

Bob constructs analogous representations and IBLTs.  Bob is able to decode all of the level 1 representations corresponding to differing documents, since the IBLT is of size $\Theta(d)$.  Bob consider all combinations of his documents and Alice's extracted representations.  For each of these pairs, he attempts to perform his side of the document exchange protocol used in \autoref{thm:de_best}.  Since the bound $k$ Alice used at this level was 2, the document exchange protocol will succeed (with high probability) whenever he pairs one of his documents with Alice's representation of a document that differs by at most 2 edits.  For those pairs differing in more than 2 edits, the protocol will generally fail, either reporting failure or by yielding an incorrect document, which he can detect using Alice's document's hash.  He then generates representations for those documents he recovered from Alice, and removes them from the level 2 IBLT so that he can recover most of Alice's documents that have at most 4 edits, and so on, until he has recovered all of Alice's documents.  The full details of the protocol and the proof of the theorem appear in \autoref{app:sosdetails}.

If the only bound on $h$ and $s$ we have is that they are both $O(n)$ (as they must be) then \autoref{thm:best_sos} uses $O(d \log n \log d + d \log^3 d)$ bits of communication, which already generally outperforms our previous best one round protocol, \autoref{thm:reduction}, which uses $O(d \log n \log(n/d))$ bits of communication.  The difference really shows when we have a tight bound on $h$, such as $h = \poly \log n$.  In such a case \autoref{thm:best_sos} uses only $O(d \log n \min(\log d, \log \log n))$ bits of communication, outperforming \autoref{thm:reduction} by at least a factor of $\Theta(\log n / \log \log n)$.  This $h$ regime is not an unnatural one, as it corresponds to the setting where our directory consists of a large number of small files.

\subsection{Speeding Up}

The main weakness of \autoref{thm:best_sos} is the running time, which is $\tilde{\Omega}(nd)$, rather than a more desirable $\tilde{O}(n)$.  The original protocol for reconciling sets of sets on which this was based has a running time of $\tilde{O}(n+\poly(d))$.  The reason that \autoref{thm:best_sos} takes more time is that the document exchange protocols take time linear in the length of the documents to determine if protocol will succeed, while the underlying set reconciliation protocol used in \cite{mitzenmacher2017reconciling} only takes time linear in the specified difference bound.  This is a bottleneck because each document exchange message is compared to many documents to see if they can decode it, but each attempt takes linear time.  In other words, if we can make a document exchange protocol that determines failure within time $O(\poly(k))$, we can reduce the computation time of \autoref{thm:best_sos} to $\tilde{O}(n+\poly(d))$.  This holds even if Bob is allowed $\tilde{O}(n+\poly(k))$ precomputation time before receiving Alice's message, since that would be a one time cost for each document.  We develop such a document exchange protocol here, but it is primarily a proof of concept as it has an increased communication cost and so the resulting directory reconciliation protocol is very rarely superior to both \autoref{thm:reduction} and \autoref{thm:best_sos} simultaneously.  Further progress is necessary to yield an improved directory reconciliation protocol, but we believe that our approach here is a valuable starting point.

Our document exchange protocol is a modification of \autoref{thm:ims}.  We sketch our changes here, and present the full details and proof in \autoref{app:fastims}.  First, we limit ourselves to edit distance, rather than block edit distance, which means that rather than having to consider all $\Theta(n)$ possible blocks as matches, we only have to consider the $k$ shifts of each corresponding block.  For each of his possible blocks, of which there are $k^2$ in the first level, $2 k^2$ in the next and so forth, Bob precomputes the encoding of that block at that, and all lower levels.  So for example the encoding of a block from the first level at the second level would be the combination of the encoding of the first half of the block and the second half of the block.  For our error correcting code, we use IBLTs, which have the advantage that they can combine two of these precomputed pieces efficiently.  This will progress only down to the level where there are $\Theta(n/k^2)$ blocks of size $\Theta(k^2)$, at which point we transmit a direct encoding of the $O(k)$ blocks of size $\Theta(k^2)$, blowing up our communication to $(k \log n \log(n/k) + k^3)$.  There are $O(\log(n/k))$ levels, so the blocks from the first level require $O(k^2 \log(n/k))$ precomputed encodings and the bottom level requires $O(n/k)$ precomputed encodings.  Each encoding takes $O(k)$ time to generate (once we have the lower level encodings and all of the hashes, which take $O(n \log (n/k))$) so the total precomputation time takes roughly $O(n \log(n/k) + k^3 \log (n/k))$.  Once we have these precomputations, we can decode in roughly $O(k^2 \log n)$ time, since at each level we just have to combine $O(k)$ precomputed encodings of size $O(k)$, and decode.  This yields the following theorem.

\begin{restatable}{theorem}{fastims} \label{thm:fastims}
There is a one round protocol for document exchange under edit distance using 
\[O((k \log(n/k) \log (n/\delta)+k^3) \lceil \log_k \log n\rceil)\] 
bits of communication, which succeeds with probability at least $1-1/\poly(k)$ if $\Delta_e(a,b) \leq k$.  Alice takes $O(n \log (n /k))$ computation time, and Bob has an 
\[O\left(n  \log(n/k) \log_n(n/\delta) + k^3 \log(n/k) \log_n(n/\delta) \lceil \log_k \log n\rceil\right)\]
time precomputation step, before receiving Alice's message.  After receiving Alice's message, Bob has an 
\[O(k^2 \log(n/k) \log_n(n/\delta) \lceil \log_k \log n\rceil)\]
time \emph{failure check} step.  If the message passes the failure check, then the protocol will succeed with probability at least $1-\delta$ (independent of $\Delta_e(a,b)$) after a final $O(n \lceil \log_k \log n\rceil)$ time step from Bob.
\end{restatable}

Using this document exchange protocol as our subroutine instead of \autoref{thm:de_best} in \autoref{thm:best_sos} yields the following directory reconciliation protocol.  

\begin{restatable}{theorem}{thmfastsos} \label{thm:fastsos}
Directory reconciliation with can be solved in one round using 
\[O(d\log s+ d \log n \log h \log \min(d,h) \log \log h + d \min(d,h)^2 \log \log h)\]
bits of communication and 
\[O(n \log h \log \min(d,h) + d^3 \log h \log \min(d, h) \log \log h+d^2 \q^2 \log n \log \log h)\]
time with probability at least $3/5$.
\end{restatable}

This degrades our communication bound in \autoref{thm:best_sos} from $\tilde{O}(d \log n)$ to $\tilde{O}(d \log^2 n + d^3)$, but speeds it up from $\tilde{O}(nd)$ to $\tilde{O}(n + d^4)$.  The proof of this theorem is also in \autoref{app:fastims}.

\subsection{Unknown $d$}
So far all of our protocols have assumed that we know $d$ (or a good bound on it) in advance.  Any of them can be extended to handle the case where $d$ is unknown by the repeated doubling method.  First we try the protocol with $d=1$ and send a hash of Alice's directory along with it.  If the protocol succeeds and the directory Bob recovers matches the hash, then we are done, otherwise we proceed to $d=2,4,8,...$ until we eventually find the right value of $d$ and succeed with good probability.  This strategy will not increase any of our asymptotic communication costs and will only increase the computation times by at most a factor of $\log d$, however all of our one round protocols become $\Theta(\log d)$ round protocols.  For some applications, it is important to minimize the number of rounds of communication, so in this section we develop a protocol that allow for unknown/unbounded values of $d$ while still using a constant number of rounds of communication.

Our key techniques here are the CGK embedding \cite{chakraborty2016streaming} and set difference estimators \cite{mitzenmacher2017reconciling}.  The CGK embedding is a randomized mapping from edit distance space to Hamming space, such that with constant probability the Hamming distance between embedded strings will be $\Omega(k)$ and $O(k^2)$, where $k$ was their original edit distance.  Set difference estimators are $O(\log n)$ bit sketches that allow the computation of a constant factor approximation to the size of the difference between two sets, with constant probability.  They can equivalently be used to estimate the Hamming distance between two strings.  By combining these tools with \autoref{thm:reduction} we get the following result.

\begin{restatable}{theorem}{multiroundnod}
Directory reconciliation with unknown $d$ can be solved in 4 rounds using 
\[O(\q \log n \log s / \log \q + \q \log n \log \max(\q, h) \log h + d^2 \log n)\]
bits of communication and 
\[O((sh + \q^2)\log n \log \max(\q, h)+d^2 \log n)\]
time with probability at least $1-1/\poly(n)$.
\end{restatable}

This theorem is inspired by the multi-round approach to reconciling sets of sets in \cite{mitzenmacher2017reconciling}.  In the first round, we estimate the number of differing documents via set difference estimators.  In the second round we perform set reconciliation on sets of hashes of the documents, to determine which among them differ.  In the third round we exchange set difference estimators for the CGK embeddings of the differing documents, to determine which documents have close edit distance to which others, and what that edit distance is.  In the final round we use that information to reconcile the differing documents by using Hamming distance sketches applied to the CGK embeddings used in the previous round.  The fact that we use the same CGK embedding for reconciliation as for the estimate is what allows us to end up with only a final $O(d^2 \log n)$ term.  The full details and proof are presented in \autoref{app:unknownd}.

\section{Conclusion}
Directory reconciliation considers the reconciliation problem in a natural practical setting.  While directory reconciliation is closely related to document exchange and document 
exchange with block edits variations, it has its own distinct features and challenges; also, while the problem is closely related to practical tools such as rsync, rsync works on a file-by-file basis that may be inefficient in some circumstances.  Theoretically, we have found a document exchange scheme with $O(k \log n)$ bits of communication to handle $k$ edits with block moves, at the expense of a number of rounds of interaction.  The natural question is whether this result can be achieved with a single round.  On the more practical side, we believe using the ``set of sets'' paradigm based on IBLTs may provide mechanisms that, besides being of theoretical interest, may also be useful in some real-world settings, where files may not be linked by file names but are otherwise closely related.

\bibliographystyle{plainurl}
\bibliography{approx_set_recon}

\appendix

\section{Reducing to Document Exchange with Block Edits}
\label{app:reduction}
We show here that directory reconciliation can be solved via a straightforward reduction to document exchange under edit distance with block moves.  

\thmreduction*

\begin{proof}
First we describe the protocol, and then we will argue its correctness.  Alice and Bob compute a $\Theta(\log n)$ bit hash of each of their documents, then sort their hashes in $O(s)$ time using radix sort.  They then concatenate all of their documents into a single document, ordering them by the sorted order of their hashes.  They choose a random $\Theta(\log n)$ bit delineation string and insert it in between each pair of documents in the concatenation.  They then engage in document exchange with their concatenated documents so that Bob recovers Alice's concatenated document.  They use \autoref{thm:ims} with $k = 2d$.  Bob then converts the concatenated document into a directory by splitting along the delineation strings, thus recovering Alice's directory.

With probability at least $1-1/\poly(n)$, no two of Alice and Bob's documents that are not equal will have equal hash values.  Assuming this is the case, \autoref{thm:ims} will allow Bob to recover Alice's concatenated document with probability at least $1-1/n$.  This is because without hash collisions, the documents will be ordered such that at most $\q$ block moves are required to order Alice's documents so that Alice and Bob's document orders correspond to the minimum edit distance matching between their documents.  After making these block moves, the edit distance between Alice's and Bob's concatenated documents would be at most $d$.  The optimal block edit distance between the original concatenated documents is therefore seen to be at most $2d$, and thus Bob can recover Alice's concatenated document using the algorithm of \autoref{thm:ims} (for $k=2d$) with probability at least $1-1/\poly(n)$.  With probability at least $1-1/\poly(n)$ the delineation string will not appear in any of the documents, so splitting Alice's concatenated document using it will produce Alice's directory.

The concatenated documents are of size $O(n+s\log(n))=O(n)$ since the directories cannot have duplicate documents (since they are sets) so $s \leq n / \log n$.  Thus, by \autoref{thm:ims} the communication cost is $O(d \log n \log(n/d))$. The running time is dominated by the $O(n \log(n/d))$ used by \autoref{thm:ims} since it only takes $O(n)$ time to construct the concatenated strings and then extract the directory.
\end{proof}

\section{Missing Proofs for the Multi-Round Document Exchange Protocol}
\label{app:multiproofs}
Here we providing the missing proofs from the lemmas used to prove that \autoref{alg:prefixprotocol} satisfies \autoref{lem:prefixprotocol}.  Recall that $X_i$ is $w(g)/w(T)$ at the beginning of phase $i$ in \textsc{MultWeightsProtocol}, and $Y_i = \ln(X_{i+1} / X_i)$.

\lemybounded*

\begin{proof}
We consider three cases, based on the value of $\eta$ in phase $i$.  First, let $\eta = 0$ and so no queries are performed and a single vertex $v \neq g$ is added to $M$, such that $w(v) > w(T) / 2$.  In this case,
\[Y_i = \ln(w(T)/(w(T)-w(V))) \geq \ln 2,\]
as desired.

Now let $\eta = 1$ so there is a single query performed in phase $i$.  Let $r$ be the query's depth.  We consider two subcases: $g \in T_{-r}$ and $g \in T_v$ for $v \in D_r$.

$Y_i$ is maximized when there are no hash collisions.  That is, $h(a_{:r}) = h(u)$ if and only if $a_{:r} = u$.  If $g \in T_{-r}$, and there are no hash collisions, then
\begin{align*}
Y_i &= \ln\left(\frac{p^{2|D_r|} w(g)}{p^{2|D_r|} w(T_{-r}) + p^{2|D_r|-2}(1-p)^2(w(T)-w(T_{-r}))} / \frac{w(g)}{w(T)} \right) \\
&= \ln\left(\frac{p^2}{(2p-1) w(T_{-r}) / w(T) + (1-p)^2} \right) \\
&\leq 2\ln\left(\frac{p}{1-p} \right).
\end{align*}

If $g \in T_v$ for $v \in D_r$ and there are no hash collisions, then
\begin{align*}
Y_i &= \ln\left(\frac{p^3}{(p^3-(1-p)^3) w(T_v)/w(T) + (1-p)(2p-1)w(T_{-r})/w(T)+(1-p)^3}\right) \\
&\leq 3\ln\left(\frac{p}{1-p} \right).
\end{align*}

$Y_i$ is minimized when every single hash collides with $h(a_{:r})$.  $h(a_{:r}) = h(u)$ for all $u \in D_r$.  If $g \in T_{-r}$, and all hashes collide, then
\begin{align*}
Y_i &= \ln\left(\frac{(1-p)^{|D_r|} w(T)}{(1-p)^{|D_r|} w(T_{-r}) + p(1-p)^{|D_r|-1}(w(T)-w(T_{-r}))}\right) \\
&= \ln\left(\frac{1-p}{p-(2p-1) w(T_{-r})/w(T)}\right) \\
&\geq -\ln\left(\frac{p}{1-p} \right).
\end{align*}

Finally we consider the case when $g \in T_v$ for $v \in D_r$, and all hashes collide.  Note that in this case we are still guaranteed that $h(a_r) = h(v)$ since $a_r = v$.
\begin{align*}
Y_i &= \ln\left(\frac{p}{p-(2p-1)w(T_{-r})/w(T)}\right) \\
&\geq 0.
\end{align*}

Finally, if $\eta = 2$ then 2 queries are performed in phase $i$.  By our previous arguments about one query, we immediately have that
\[-2\ln\left(\frac{p}{1-p}\right) \leq Y_i \leq 6\ln\left(\frac{p}{1-p}\right).\]
\end{proof}

\lemymean*

\begin{proof}
If $\eta = 0$ then by \autoref{lem:ybounded}, $Y_i  \geq \ln 2$, so clearly $\E[Y_i] \geq 0$.

If $\eta = 1$, then phase $i$ performs a single query at depth $r$.  Let $U_i$ be a random variable taking the value of the set of nodes $u \in D_r$ such that $h(u) \neq h(a_{:r})$.  

First, let $g \in T_{-r}$.  For each $u \in D_r$, $\Pr[h(u)\neq h(a_{:r})] \geq p$.  Thus,
\begin{align}
E[Y_i] &= E[\ln(X_{i+1}/X_i)] \nonumber \\
&\geq \ln(1/E[X_i/X_{i+1}]) \; \text{ (Jensen's inequality)} \nonumber \\
&= \ln\left(1/E\left[\frac{p^2(1-p)w(T_{-r})+(1-p)^3\sum_{u\in U_i}w(T_u) + p^3\sum_{u \in D_r \setminus U_i}w(T_u)}{p^2(1-p)w(T)}\right]\right) \nonumber \\
&= \ln\left(\frac{p^2(1-p)w(T)}{p^2(1-p)w(T_{-r})+(1-p)^3E\left[\sum_{u\in U_i}w(T_u)\right] + p^3E\left[\sum_{u \in D_r \setminus U_i}w(T_u)\right]}\right) \nonumber \\
&\geq \ln\left(\frac{p^2(1-p)w(T)}{p^2(1-p)w(T_{-r})+p(1-p)^3\sum_{u\in D_r}w(T_u) + p^3(1-p)\sum_{u \in D_r}w(T_u)}\right) \\
&= \ln\left(\frac{p^2(1-p)w(T)}{p^2(1-p)w(T_{-r})+(p^3(1-p)+p(1-p)^3)(w(T)-w(T_{-r}))}\right) \nonumber \\
&= \ln\left(\frac{p}{(1-p)(2p-1)w(T_{-r})/w(T)+1-2p(1-p)}\right) \label{eq:topg}\\
&\geq 0 \; \text{ (since $w(T_{-r}) \leq w(T)$)},\nonumber
\end{align}
where the second inequality used linearity of expectations and then convexity.

For informative queries, we have that $w(T_{-r}) \leq (3/4)w(T)$.  Plugging this into \autoref{eq:topg} we get\footnote{We note that this bound approaches approaches 0 as $p$ approaches 1, which is an unintuitive artifact of our analysis.  However, since we choose $p=2/3$, this is bounded away from 0.}
\[E[Y_i] \geq \ln(4p/(2p^2+p+1)).\]

Now we address the case when $g \in T_v$ for $v \in D_r$.
\begin{align}
&E[Y_i] \geq \ln(1/E[X_i/X_{i+1}]) \nonumber \\
&= \ln\left(1/E\left[\frac{p^3w(T_v)+ p(1-p)w(T_{-r})+(1-p)^3\sum_{u\in U_i}w(T_u) + p^3\sum_{u \in D_r \setminus (U_i \cup \{v\})}w(T_u)}{p^3w(T)}\right]\right) \nonumber \\
&= \ln\left(\frac{p^3w(T)}{p^3w(T_v)+ p(1-p)w(T_{-r})+(1-p)^3\E\left[\sum_{u\in U_i}w(T_u)\right] + p^3\E\left[\sum_{u \in D_r \setminus (U_i \cup \{v\})}w(T_u)\right]}\right) \nonumber \\
&\geq \ln\left(\frac{p^3w(T)}{p^3w(T_v)+ p(1-p)w(T_{-r})+(p^3(1-p)+p(1-p)^3)(w(T)-w(T_v)-w(T_{-r}))}\right) \nonumber \\
&= \ln\left(\frac{p^2}{(2p-1)(1-p(1-p))w(T_v)/w(T)+2p(1-p)^2w(T_{-r})/w(T)+(1-p)(1-2p(1-p))}\right) \label{eq:botg}.
\end{align}

Since $p \geq 2/3$, $(2p-1)(1-p(1-p)) > 2p(1-p)^2$.  Thus, using only the constraint that $w(T_v)+w(T_{-r})\leq w(T)$,  \autoref{eq:botg} is minimized when $w(T_v) = w(T)$ and $w(T_{-r})=0$, which yields
\[\E[Y_i] \geq 0.\]
For informative queries, we also have the constraint that $w(T_v) \leq (3/4)w(T)$, in which case \autoref{eq:botg} is minimized when $w(T_v) = (3/4)w(T)$ and $w(T_{-r}) = (1/4)w(T)$.  This gives us
\[\E[Y_i] \geq \ln(4p^2/(3p^2-p+1)).\]
Since $p \geq 2/3$, $\ln(4p^2/(3p^2-p+1)) > \ln(4p/(2p^2+p+1))$.

Finally we consider cases where $\eta = 2$ and thus two queries are performed.  Since $E[Y_i] \geq 0$ for one query, then by linearity of expectations $E[Y_i] \geq 0$ for two queries as well.  Furthermore, if phase $i$ is informative then at least one of the queries is informative so by linearity of expectations and our analysis of the second queries cases, $E[Y_i] \geq \ln(4p/(2p^2+p+1)).$
\end{proof}

\lemmwpprob*

\begin{proof}
In order for $M$ to not include $g$, at no point in the protocol can $w(g) > w(T)/2$.  In particular, after the final phase we must have $w(g) \leq w(T)/2$.  Since initially $w(T)=|T|$ and $w(g) = 1$, this means that $X_{t+1}/X_1 \leq |T|/2$.  We will prove the lemma by showing that if $g$ is never added to $M$, then $\Pr[X_{t+1}/X_1 \leq |T|/2] \leq \delta$.

Let $Z_i = 1/4 + Y_i/(8 \alpha)$, where $\alpha = \ln(p/(1-p))$.  Let $H \subseteq [t]$ be the set of phases in which a node is added to $M$.  By \autoref{lem:ybounded}, for $i \notin H$, $0 \leq Z_i \leq 1$ and for $i \in H$, $Z_i \geq 1/4 + \ln(2)/(8 \alpha)$.  By \autoref{lem:ymean}, $\E[Z_i] \geq 1/4$ and if an informative query is performed in phase $i$, then $\E[Z_i] \geq 1/4 + \beta/(8\alpha)$, for $\beta = \ln(4p/(2p^2+p+1))$.  By \autoref{lem:informative}, out of $t$ queries, at least $(2t - 8\log_{3/2} \tau)/7$ of them are informative.  Assuming $t \geq 32 \log_{3/2} \tau$, 
\[\mu = \E\left[\sum_{i \in [t] \setminus H} Z_i\right] \geq  t(1+\beta / (8 \alpha))/4-|H|(1+\beta/(2\alpha))/4.\]

Putting these pieces together we have
\begin{align}
&\Pr\left[X_{t+1}/X_1 \leq |T|/2\right] = \Pr\left[\sum_{i=1}^t Y_i \leq \ln(|T|/2)\right] \nonumber \\
&= \Pr\left[\sum_{i=1}^t Z_i \leq \frac{\ln(|T|/2)+t\alpha}{8\alpha}\right] \nonumber \\
&\leq \Pr\left[\sum_{i\in[t]\setminus H} Z_i \leq \frac{\ln(|T|/2)+t\alpha-|H|(\ln 2+2\alpha)}{8\alpha}\right] \label{eq:mwpprobbound} \\
&\leq \Pr\left[\sum_{i\in[t]\setminus H} Z_i \leq \left(1-\epsilon\right)\mu\right] \; \text{ (where $\epsilon = \frac{t (4\alpha + \beta) + |H|(\ln 16 - 4\beta) - 4 \ln(|T|/2)}{t(8\alpha + \beta)-|H|(8\alpha + 4\beta)}$)} \nonumber \\
&\leq e^{-\epsilon^2 \mu /2} \; \text{ (Chernoff bound, assuming $\epsilon \in [0,1]$)} \nonumber \\
&\leq e^{-\frac{(t(4\alpha+\beta)+8|H|(\ln 2 - \beta))^2}{128\alpha(t(8\alpha+\beta)-4|H|(2\alpha+\beta))}} \; \text{ (assuming $t \geq 8 \ln(|T|/2)/(4\alpha+\beta)$)} \nonumber \\
&\leq e^{-\frac{t(4\alpha+\beta)^2}{128\alpha(8\alpha+\beta)}} \; \text{ (maximized at $|H|=0$)} \nonumber \\
&\leq \delta\; \text{ (assuming $t \geq \ln(1/\delta)(128\alpha(8\alpha+\beta)/(4\alpha+\beta)^2)$)}. \nonumber 
\end{align}

Now let's address our assumption that $\epsilon \in [0,1]$.  If $\epsilon > 1$ then 
\[\Pr\left[\sum_{i\in[t]\setminus H} Z_i \leq \left(1-\epsilon\right)\mu\right] \leq \Pr\left[\sum_{i\in[t]\setminus H} Z_i < 0\right] = 0 < \delta.\]
For $p=2/3$, $\ln 16 - 4 \beta > 0$ so with our assumption that $t \geq 8 \ln(|T|/2)/(4\alpha+\beta)$, the numerator of $\epsilon$ must be $> 0$, thus in order to have $\epsilon < 0$ we must have $|H| > t(8\alpha + \beta)/(8\alpha+4\beta)$.  Plugging this into \autoref{eq:mwpprobbound} we have
\begin{align*}
\Pr&\left[\sum_{i\in[t]\setminus H} Z_i \leq \frac{\ln(|T|/2)+t\alpha-|H|(\ln 2+2\alpha)}{8\alpha}\right] \\
&\leq \Pr\left[\sum_{i\in[t]\setminus H} Z_i \leq \frac{\ln(|T|/2)+t\left(\alpha-\frac{(8\alpha+\beta)(\log 2 + 2\alpha)}{8\alpha+4\beta}\right)}{8\alpha}\right] \\
&\leq \Pr\left[\sum_{i\in[t]\setminus H} Z_i \leq \ln(|T|/2)\frac{1+\frac{8}{4\alpha+\beta}\left(\alpha-\frac{(8\alpha+\beta)(\log 2 + 2\alpha)}{8\alpha+4\beta}\right)}{8\alpha}\right] \\
&\leq \Pr\left[\sum_{i\in[t]\setminus H} Z_i < 0\right] \; \text{ (plugging in $p = 2/3$)}\\
&= 0 < \delta.
\end{align*}

Putting all of our assumptions about $t$ together, we have
\[t \geq 32 \log_{3/2} \tau + \frac{8}{4\alpha+\beta} \ln(|T|/2) + \frac{128\alpha(8\alpha+\beta)}{(4\alpha+\beta)^2}\ln(1/\delta).\]
Plugging in $p=2/3$, this is satisfied when
\[t \geq 79 \ln \tau + 3 \ln |T| + 63 \ln(1/\delta)\]
as desired.
\end{proof}

\section{Set of Sets Based Protocol Details}
\label{app:sosdetails}

Here we present the details of our protocol for directory reconciliation based on a protocol for reconciling set of sets of \cite{mitzenmacher2017reconciling}.  First we present a general protocol in which one can plug in any one round document exchange protocol, and then we go on to apply it with specific protocols.  Let $n_j$ be the size of $j$th largest document in the union of Alice and Bob's directories.  $h \geq n_1 \geq n_2 \geq \ldots \geq n_m$, for $m \leq 2s$ and $\sum_{j=1}^m n_j \leq 2n$.  

\begin{lemma}
\label{lem:ibltofdocs}
Given a one-round document exchange protocol with time $g(k,n')=O(\poly(k,n'))$, communication cost $f(k,n')=O(\poly(k,n'))$, and success probability at least $1 - 1/(100k)$, directory reconciliation can be solved in one round using 
\[O\left(d\log s +  d \log \log \min(d,h) + \sum_{i=1}^{\log \min(d,h)} (d/2^i) g(2^i, h)\right)\]
bits of communication and 
\[O\left(\sum_{i=1}^{\log \min(d,h)} \left(\sum_{j=1}^m f(2^i, n_j)+\sum_{j=1}^{\min(\q, d/2^i)} (\q + \min(\q, d/2^i) - 2j) f(2^i,n_j)\right)\right)\] 
time with probability at least $2/3$.
\end{lemma}

\begin{algorithm}[h]
\caption{Cascading IBLTs of Document Exchange Protocols}
\label{alg:ibltofdocs}
\begin{enumerate}
\item For $i = 1, \ldots, t = \log_2 \min(d,h)$, Alice creates a (document exchange message with $k=\Theta(2^i)$, $\Theta(\log (st))$ bit hash) document encoding for each of her documents and inserts it into a $\Theta(d / 2^i)$ cell IBLT $T_i$.
\item If $t = \log_2 h$, Alice creates a $\Theta(d/h)$ cell IBLT $T_*$ and inserts a $\Theta(h)$ bit encoding of each of her documents into it.
\item Alice sends $T_1,\ldots,T_t$ and $T_*$ to Bob.
\item Bob deletes (document exchange message with $k=O(1)$, hash) encodings of each of his documents from $T_1$, and then extracts all of the different document encodings from it.  He uses the hashes of his extracted documents to recover $D_B$, the set of his documents that differ from any of Alice's.  
\item Bob tries performing the document exchange protocol using every combination of one Alice's extracted document encodings and one of his documents in $D_B$, trying to recover Alice's documents.  Each time the protocol succeeds, and the resulting document matches Alice's hash, he inserts the recovered document into the set $D_A$.  Going forward, he will recover more and more of Alice's documents and $D_A$ will be the set he has recovered so far.
\item For each $i = 2,\ldots,t$, Bob performs the following procedure.  He first deletes the level $i$ document encoding of each of his documents from $T_i$, except for those in $D_B$. He also deletes the level $i$ document encoding of each document in $D_A$ from $T_i$.  He then decodes $T_i$ and extracts all of the different document encodings, which correspond exactly to Alice's differing documents that aren't yet in $D_A$.  He tries to to decode each of Alice's extracted document encodings with each of the documents in $D_B$, adding Alice's documents that he recovers to $D_A$.
\item If Bob received $T_*$, he deletes all of his documents from it.  He also deletes each document in $D_A$ from it.  He then decodes $T_*$ and adds all of the decoded documents to $D_A$.
\item Bob deletes $D_B$ from his directory, and adds $D_A$.
\end{enumerate}
\end{algorithm}

\begin{proof}
We analyze the protocol given in Algorithm \ref{alg:ibltofdocs}.   Our proof is essentially identical to that of \cite{mitzenmacher2017reconciling}.  Going forward we will condition on the event that there are no collisions among the $\Theta(\log (st))$ bit hashes in the document encodings.  There are at most $2s$ documents per round that can collide, so union bounding over the $t$ rounds we have no collisions with probability at least $1-4s^2t/\poly(st)\geq 1-1/30$.  

Let us divide Alice's documents into groups according to how many edits they differ by under the minimum difference matching.  $S_j$ is the set of Alice's documents whose edit distance with its match is in $[2^{j-1}, 2^j - 1]$.  First, observe that every one of Alice's differing documents is included in some $S_j$ for $j \leq t$.  Second, observe that $|S_j| \leq d / 2^{j-1}$ since the total number of edits is at most $d$.

Consider the IBLT $T_i$. Let $Y_i$ be the event that IBLT $T_i$ successfully decodes.  Conditioned on $Y_i$, when processing to match up the documents within $T_i$, in expectation Bob fails to recover at most $\frac{1}{100 \cdot 2^i}$ of Alice's documents from $\cup_{j=1}^i S_j$ that he has not yet recovered.  By Markov's inequality, Bob recovers fewer than 9/10 of the Alice's documents in $\cup_{j=1}^i S_j$ that he has not already decoded with probability at most $\frac{1}{10 \cdot 2^i}$.  We use $X_i$ to refer to the event that processing $T_i$ results in Bob recovering at least $9/10$ of $\cup_{j=1}^i S_j$ that he had not previously recovered, so we have argued that
\[\Pr[X_i | Y_i] \geq 1 - \frac{1}{10 \cdot 2^i}.\]

Since there are at most $2d$ differing documents in $T_1$, we can choose the IBLT's parameters so that $Y_1$ occurs with probability at least $1 - \frac{2}{10d}$.  For $i > 1$, conditioned on $\cap_{j=1}^{i-1}X_j$, the number of Alice's documents left to be recovered after $T_i$ is processed is at most
\begin{align*}
\sum_{j=i+1}^t &|S_j| + \sum_{j=1}^i |S_j| 10^{j-i-1} \\
&\leq \sum_{j=i+1}^t d / 2^{j-1} + \sum_{j=1}^i d 10^{j-i-1} / 2^{j-1} \\
&\leq d / 2^{i-1} + d / 10^i \sum_{j=1}^i 5^{j-1} \\
&\leq d / 2^{i-1} + d / 2^{i+2} = (9 / 4) (d / 2^i).
\end{align*}
Since $T_i$ has $\Theta(d / 2^i)$ cells, we can choose the constant factors in the order notation so that $Y_i$ occurs with probability at least $1 - \frac{2^i}{10 d}$ conditioned on $\cap_{j=1}^{i-1}X_j$.

If $t < \log_2 h$, and therefore there is no $T_*$, Bob successfully recovers all of Alice's documents so long as all $X_i$ and $Y_i$ occur.  The probability of this is
\begin{align*}
\Pr&[\cap_{i=1}^t (X_i \cap Y_i)] \\
&= \Pr[Y_1] \Pr[X_1 | Y_1] \ldots \Pr\left[Y_t | \cap_{j=1}^{t-1}X_j\right] \Pr[X_t | Y_t] \\
&= \Pr[Y_1] \prod_{i=2}^t \Pr\left[Y_i | \cap_{j=1}^{i-1}X_j\right] \prod_{i=1}^t \Pr[X_i | Y_i]  \\
&= \prod_{i=1}^t \left(1 - \frac{2^i}{10d}\right) \prod_{i=1}^t \left(1 - \frac{1}{10 \cdot 2^i}\right) \\
&\geq 1 - \sum_{i=1}^t \left(\frac{2^i}{10d} + \frac{1}{10 \cdot 2^i}\right) \\
&\geq 4/5.
\end{align*}
If $t = \log_2 h$, then the protocol will succeed so long as $T_*$ successfully decodes.  $T_*$ has $\Theta(d/h)$ cells and if all $X_i$ and $Y_i$ occur then there are at most $(9/4)d/h$ elements to extract from $T_*$, so we can choose the constants such that $T_*$ decodes with probability at least $1/10$.  We have thus proved that by the end of the procedure, Bob recovers Alice's directory with probability at least $4/5-1/10-1/30=2/3$.

The time for Alice and Bob to construct their document encodings and insert or delete them from the $T_{i}$ is \[O\left(\sum_{i=1}^t \sum_{j=1}^m f(2^i, n_j)\right).\]
The remaining time is what Bob takes to attempt to decode Alice's document encodings.  When processing $T_i$, Bob extracts $O(\min(\q, d/2^i))$ of Alice's document encodings, and compare each one against each of his $O(\q)$ differing documents' encodings.  Each document encoding has $k = \Theta(2^i)$, so the total processing time is at most 
\[O\left(\sum_{j=1}^{\min(\q, d/2^i)} \sum_{\ell=1}^{\q} f(2^i,\max(n_j,n_\ell))\right) = O\left(\sum_{j=1}^{\min(\q, d/2^i)} (\q + \min(\q, d/2^i) - 2j) f(2^i,n_j)\right).\]
Summing over $i$, we get
\[O\left(\sum_{i=1}^t \sum_{j=1}^{\min(\q, d/2^i)} (\q + \min(\q, d/2^i) - 2j) f(2^i,n_j)\right).\]

The communication cost of transmitting $T_{1},\ldots,T_{t}$ and $T_*$ is
\begin{align*}
O&\left(\sum_{i=1}^t (d/2^i) \cdot (\log (st) + g(2^i, h))+d\right) \\
&= O\left(d \log s + d \log \log \min(d,h) + \sum_{i=1}^t (d/2^i)g(2^i, h)\right).
\end{align*}
\end{proof}

As a first attempt, we combine this lemma with \autoref{thm:ims}.  Since \autoref{thm:ims} works under edit distance with block moves, this protocol has the advantage that it allows block moves within the files as edit operations, just as \autoref{thm:reduction} and \autoref{cor:multiround} do.

\begin{theorem} \label{thm:ibltofims}
Directory reconciliation can be solved in one round using 
\[O(d\log s+ d\log^2 h \log \min(d,h))\]
bits of communication and 
\[O(n d \log h +n \q \log h \log d)\]
time with probability at least $2/3$.
\end{theorem}

\begin{proof}
We use \autoref{lem:ibltofdocs} with \autoref{thm:ims} as our document exchange protocol, giving us $f(k,n')=n'\log n'$ and  $g(k,n')=k\log n' \log (n'/k)$.  Our communication bound is
\begin{align*}
O&\left(d\log s +  d \log \log \min(d,h) + \sum_{i=1}^{\log \min(d,h)} (d/2^i)g(2^i, h)\right) \\
&= O\left(d\log s +  d \log \log \min(d,h) + \sum_{i=1}^{\log \min(d,h)} d \log h \log(h/2^i)\right) \\
&= O(d\log s+ d\log^2 h \log \min(d,h)).
\end{align*}

We bound the time in two pieces.  First,
\begin{align*}
O\left(\sum_{i=1}^{\log \min(d,h)} \sum_{j=1}^m f(2^i, n_j)\right) &= O\left(\sum_{i=1}^{\log \min(d,h)} \sum_{j=1}^m n_j \log n_j\right) \\
&= O\left(n \log h \log \min(d,h)\right).
\end{align*}

Finally,
\begin{align*}
O&\left(\sum_{i=1}^{\log \min(d,h)} \sum_{j=1}^{\min(\q, d/2^i)} (\q + \min(\q, d/2^i) - 2j) f(2^i,n_j)\right) \\
&= O\left(\sum_{i=1}^{\log d} \sum_{j=1}^{\min(\q, d/2^i)} (\q + \min(\q, d/2^i) - 2j) n_j \log n_j \right) \\
&= O\left(\sum_{i=1}^{\log (d/\q)} \sum_{j=1}^{d/2^i} (\q + d/2^i - 2j) n_j \log n_j +\sum_{i=\log (d/\q)}^{\log d} \sum_{j=1}^{\q} (2\q - 2j) n_j \log n_j \right)\\
&= O\left(\sum_{i=1}^{\log (d/\q)} (\q + d/2^i) n \log h +\sum_{i=\log (d/\q)}^{\log d} 2\q n \log h \right)\\
&= O(n \log h (d + \q\log d)).
\end{align*}
\end{proof}

Now we instead use the state-of-the-art one round document exchange protocol \autoref{thm:de_best} to achieve the following result.

\thmbestsos*
\begin{proof}

We use \autoref{lem:ibltofdocs} with \autoref{thm:de_best} as our document exchange protocol when our chosen $k < h^\epsilon$.  In this case $f(k,n')=n'(\log k + \log \log n')$ and $g(k,n')=k\log n' +k\log^2 k$.  When $k \geq h^\epsilon$ we use \autoref{thm:ims} as in \autoref{thm:ibltofims}.  Our communication cost is
\begin{align*}
O&\left(d\log s +  d \log \log \min(d,h) + \sum_{i=1}^{\log \min(d,h)} (d/2^i) g(2^i, h)\right) \\
&= O\left(d\log s +  d \log \log \min(d,h) + \sum_{i=1}^{\log \min(d,h)} d( \log h + i^2)\right) \\
&= O(d \log s + d \log h \log \min(d,h) + d \log^3 \min(d,h)).
\end{align*}

We bound the running time just as in \autoref{thm:ibltofims}, but we consider two cases.  First, if $d < h^\epsilon$ then
\begin{align*}
O\left(\sum_{i=1}^{\log \min(d,h)} \sum_{j=1}^m f(2^i, n_j)\right) &= O\left(\sum_{i=1}^{\log d} \sum_{j=1}^m n_j(i + \log \log h)\right) \\
&= O\left(n \log^2 d + n \log d \log \log h\right).
\end{align*} 
Furthermore,
\begin{align*}
O&\left(\sum_{i=1}^{\log \min(d,h)} \sum_{j=1}^{\min(\q, d/2^i)} (\q + \min(\q, d/2^i) - 2j) f(2^i,n_j)\right) \\
&= O\left(\sum_{i=1}^{\log (d/\q)} (\q + d/2^i) n (i + \log \log h) +\sum_{i=\log (d/\q)}^{\log d} 2\q n (i + \log \log h) \right)\\
&= O(n d \log \log h + n \q\log d \log \log h +n \q\log^2d).
\end{align*}

Now consider $d \geq h^\epsilon$.  There are a several subcases here depending on the relative values of $d/\q,d,h^\epsilon$ and $h$, but they all yield the same bound.  We only present here the case when $h^\epsilon \leq d/\q \leq \min(d,h)$.
\begin{align*}
O\left(\sum_{i=1}^{\log \min(d,h)} \sum_{j=1}^m f(2^i, n_j)\right) &= O\left(\sum_{i=1}^{\epsilon \log h} \sum_{j=1}^m n_j(i + \log \log h)+\sum_{i=\epsilon \log h}^{\log \min(d,h)} \sum_{j=1}^m n_j \log h\right) \\
&= O\left(n \log^2 h\right),
\end{align*}
and
\begin{align*}
O&\left(\sum_{i=1}^{\log \min(d,h)} \sum_{j=1}^{\min(\q, d/2^i)} (\q + \min(\q, d/2^i) - 2j) f(2^i,n_j)\right) \\
&= O\left(\sum_{i=1}^{e\log h} (\q + d/2^i) n (i + \log \log h)+\sum_{i=e \log h}^{\log (d/\q)} (\q + d/2^i) n \log h +\sum_{i=\log (d/\q)}^{\log \min(d,h)} 2\q n \log h \right)\\
&= O(n d \log \log h +n \q\log^2 h).
\end{align*}
Combining the cases we get a total time bound of
\[O(nd \log \log h + n \q \log^2 \min(d, h) + n \q \log \min(d,h) \log \log h).\]
\end{proof}

\section{Details of Faster Decoding Document Exchange Protocol}
\label{app:fastims}

Here we develop a protocol optimized for detecting whether it is ultimately going to fail in time $o(n)$.  We then combine this with \autoref{lem:ibltofdocs} to achieve a fast directory reconciliation protocol.

\fastims*

\begin{proof}
First we detail the protocol, except for some replication for probability amplification and some computational optimization, then prove its correctness afterwards.

\begin{itemize}
\item Alice's step:
\begin{enumerate}
\item For $i \in [\lceil \log_2(n/k^3) \rceil]$, Alice divides her string into $2^ik$ blocks $\bar{a}_1, \bar{a}_2,\ldots,\bar{a}_{2^ik}$ of length $n/(2^ik)$.  For each $j \in [2^ik]$, she computes $x_{i,j}$, a $\Theta(\log(n/\delta))$ bit rolling hash of $\bar{a}_j=a_{(j-1)n/(2^ik)+1,jn/(2^ik)}$.  She creates $T_i$, a $\Theta(k)$-cell IBLT.  For $j \in [2^ik]$ she inserts the pair ($j$,$x_{i,j}$) as a key into $T_i$.
\item Alice creates a $\Theta(k)$-cell IBLT $T^*$, with both keys and values, which uses the same hash functions as $T_{\lceil \log_2(n/k^3) \rceil}$.  She then divides her string into $n/k^2$ blocks $\bar{a}_1, \bar{a}_2,\ldots,\bar{a}_{k^2}$ of length $k^2$, and for $j \in [n/k^2]$ she inserts into $T^*$ the key ($j$,$x_{\lceil \log_2(n/k^3) \rceil,j}$) with the corresponding value $\bar{a}_j$.
\item She sends $T^*$ and all of the $T_i$s to Bob.
\end{enumerate}
\item Bob's precomputation:
\begin{enumerate}
\item For $i \in [\lceil \log_2(n/k^3) \rceil]$ and $j' \in [n-n/(2^ik)+1]$ Bob computes $y_{i,j'}$, a $\Theta(\log (n/\delta))$ bit rolling hash of $b_{j':j'+n/(2^ik)-1}$.
\item For $i \in [\lceil \log_2(n/k^3) \rceil]$ and $j \in [2^ik]$, he creates a hash table $H_{i,j}$.  
\item For $i \in [\lceil \log_2(n/k^3) \rceil], j \in [2^ik]$, and $m \in \{-k,-k+1,\ldots,k\}$, let $j' = (j-1)n/(2^ik)+1+m$.  Bob inserts $y_{i,j'}$ into $H_{i,j}$ as a key, with $m$ as its corresponding value.  If a key was previously inserted into a table, pick between the values arbitrarily.
\item For $i \in [\lceil \log_2(n/k^3) \rceil], j \in [2^ik], m \in \{-k,-k+1,\ldots,k\}$, and $\ell \in \{i,\ldots,\lceil \log_2(n/k^3) \rceil\}$: Bob creates a $\Theta(k)$ cell IBLT $T_{i,j',m,\ell}$.  Let $j'=(j-1)n/(2^ik)+1+m$. He constructs $T_{i,j,m,\ell}$ so that if $b_{j':j'+n/(2^ik)-1} = \bar{a}_{j}$, $T_{i,j,m,\ell}$ will equal $\bar{a}_{j}$'s contribution to $T_\ell$.  This means that if he divides $b_{j':j'+n/(2^ik)-1}$ into $2^{\ell-i}$ blocks $\bar{b}_1,\ldots,\bar{b}_{2^{\ell-i}}$ of size $n/{2^\ell k}$, $T_{i,j,m,\ell}$ contains the pairs ($2^{\ell-i}(j-1)+\iota$,$\Theta(\log (n/\delta))$ bit rolling hash of $\bar{b}_\iota$), for $\iota \in [2^{\ell-i}]$.

\end{enumerate}
\item After Bob receives Alice's message:
\begin{enumerate}
\item Bob initializes an empty list $L$, which will consist of tuples indicating the parts of $a$ that Bob has recovered, and what parts of $b$ they match.
\item For $\ell \in [\lceil \log_2(n/k^3) \rceil]$, for each $(i,j,m)$ tuple in $L$, Bob subtracts $T_{i,j,m,\ell}$ from $T_\ell$.  Bob then attempts to decode $T_\ell$.  If it fails, he reports failure and terminates.  Otherwise, if $\ell < \lceil \log_2(n/k^3) \rceil$, then for each $(j,x_{\ell,j})$ pair that Bob extracts from $T_i$, he checks if $x_{\ell,j}$ is in $H_{\ell,j}$, and if it is he extracts the corresponding value $m$ and adds the tuple $(\ell,j,m)$ to $L$. (This is the end of the failure check phase of the protocol.)

\item For each $(i,j,m)$ tuple in $L$, let $j'= (j-1)n/(2^ik)+1+m$.  Bob divides $b_{j':j'+n/(2^ik)-1}$ into $n/(k^3 2^i)$ blocks $\bar{b}_1,\ldots,\bar{b}_{n/(k^3 2^i)}$ of size $k^2$.  For $\iota \in [n/{k^3 2^i}]$, he deletes from $T^*$ the key ($(j-1)n/(k^3 2^i)+\iota$,$\Theta(\log (n/\delta))$ bit rolling hash of $\bar{b}_\iota$) with value $\bar{b}_\iota$.  Bob then decodes $T^*$.
\item Bob creates his output string $a'$ as follows.  For each $(i,j,m)$ tuple in $L$, Bob lets $a'_{(j-1)n/(2^ik)+1,jn/(2^ik)} = b_{j':j'+n/(2^ik)-1}$, where $j' = (j-1)n/(2^ik)+1+m$.  For each key $(j,x_{i,j})$ and value $\bar{a}_j$ that Bob extracts from $T^*$, Bob lets $a'_{(j-1)k^2+1,jk^2} = \bar{a}_j$.
\end{enumerate}
\end{itemize}

We now argue that if none of the hashes collide and none of the IBLTs fail, then our protocol succeeds ($a' = a$).  This protocol operates exactly as \autoref{thm:ims} (the IMS sketch of \cite{irmak2005improved}) does except in four ways: it uses IBLTs for its systematic error correcting codes, it stops using hashes once blocks are size $\Theta(k^2)$ instead of $\Theta(\log n)$, it has a slightly different way of transmitting the encoded plain text in $T^*$ at the bottom level, and Bob has a different mechanism for deleting what he has recovered so far from Alice's codes.  We argue that these four changes still make the output consistent with that of \autoref{thm:ims}.  

IBLTs (when used to represent a vector rather than a set by including the index of each item in the pair) do indeed fulfill the requisite criteria for a systematic error correcting code here.  Stopping at blocks of size $\Theta(k^2)$ will also not affect the result, since stopping at any level, so long as at that level we directly encode the blocks rather than just the hashes, will not affect correctness.  The normal mechanism for \autoref{thm:ims} to encode the blocks at the bottom level (with IBLTs as codes) would be to make an IBLT independent of the $T_i$s and insert the the blocks paired with their indices $(j,\bar{a}_j)$ into $T^*$ as keys, without values.  Assuming no hash collisions, the same information is present in $T^*$ as we have constructed it, and assuming all of the IBLTs decode, that information will still be recovered and usable in the same way.

What remains is to argue that the way Bob uses his precomputed data to decode Alice's codes is consistent with \autoref{thm:ims}.  Suppose that Bob did not precompute anything and was just executing \autoref{thm:ims} using only the first three of our differences.  In this case, Bob's protocol after receiving Alice's message would be:
\begin{enumerate}
\item Bob initializes an empty list $L$, which will consist of tuples indicating the parts of $a$ that Bob has recovered, and what parts of $b$ they match.
\item For $\ell \in [\lceil \log_2(n/k^3) \rceil-1]$:
\begin{itemize}
\item For each $(i,j',j)$ tuple in $L$, he divides $b_{j':j'+n/(2^ik)-1}$ into $2^{\ell-i}$ blocks $\bar{b}_1,\ldots,\bar{b}_{2^{\ell-i}}$ of size $n/{2^\ell k}$.  For $\iota \in [n/{2^\ell k}]$, he deletes the pair ($2^{\ell-i}(j-1)+\iota$,$\Theta(\log (n/\delta))$ bit rolling hash of $\bar{b}_\iota$) from $T_\ell$.
\item  Bob attempts to decode $T_\ell$.  If it fails, he reports failure and terminates.  Otherwise, for each $(j,x_{\ell,j})$ pair that Bob extracts from $T_i$, he checks if there exists a $j'$ for which the $\Theta(\log (n/\delta))$ bit rolling hash of $b_{j':j'+n/(2^\ell k) -1}$ is equal to $x_{\ell,j}$. If such a $j'$ does exist, then he adds the tuple $(\ell,j',j)$ to $L$.
\end{itemize} 
\item For each $(i,j',j)$ tuple in $L$, Bob divides $b_{j':j'+n/(2^ik)-1}$ into $n/(k^3 2^i)$ blocks $\bar{b}_1,\ldots,\bar{b}_{n/(k^3 2^i)}$ of size $k^2$.  For $\iota \in [n/{k^3 2^i}]$, he deletes from $T^*$ the key ($(j-1)n/(k^3 2^i)+\iota$,$\Theta(\log (n/\delta))$ bit rolling hash of $\bar{b}_\iota$) with value $\bar{b}_\iota$.  
\item Bob attempts to decode $T^*$.  If it fails, he reports failure and ends the protocol.
\item Bob creates his output string $a'$ as follows.  For each $(i,j',j)$ tuple in $L$, Bob lets $a'_{(j-1)n/(2^ik)+1,jn/(2^ik)} = b_{j':j'+n/(2^ik)-1}$.  For each key $(j,x_{i,j})$ and value $\bar{a}_j$ that Bob extracts from $T^*$, Bob lets $a'_{(j-1)k^2+1,jk^2} = \bar{a}_j$.
\end{enumerate}
There are two differences from our protocol here.  The first is that Bob uses a precomputed hash table to determine if the hash $x_{\ell,j}$ matches the hash of a substring of $b$.  The way our protocol does this, it only checks in the table for Bob's substrings within a distance $k$ of Alice's substring.  That is, $x_{\ell,j}$ is the hash of $a_{(j-1)n/(2^\ell k)+1:jn/(2^\ell k)}$, and Bob only checks his substrings with starting index in $[(j-1)n/(2^\ell k)+1-k,(j-1)n/(2^\ell k)+1+k]$.  This is valid, because our protocol only assumes $k$ is a bound on the edit distance, not a bound on the edit distance with block moves, as is assumed in \autoref{thm:ims}.  As a result, it suffices to only look for matches of within these $k$ indices of Alice's substring.

The second difference is how we delete the known pieces of $a$ from $T_\ell$.  In the above version of \autoref{thm:ims}, we iterate through each known piece and compute the hash of each of block in that piece and delete it from $T_\ell$.  Equivalently, we could take all of those hashes and add them to a new IBLT $T'$, then subtract $T'$ from $T_\ell$.  $T'$ is exactly equal to the precomputed $T_{i,j,m,\ell}$ which our protocol subtracts from $T_\ell$, thus the two protocols are consistent.

Now we consider the failure probability.  Each of the $T_i$s decodes with probability $1-1/\poly(k)$.  We replicate the IBLTs $\Theta(\lceil \log_k \log n\rceil)$ times, so that after union bounding over the $O(\lceil \log(n/k^3) \rceil)$ IBLTs we attempt to decode, we still succeed with probability $1-1/\poly(k)$.  Since $T^*$ uses the same keys and hash functions as $T_{\lceil \log_2(n/k^3) \rceil}$, if $T_{\lceil \log_2(n/k^3) \rceil}$ decodes then $T^*$ will with probability 1.  Thus, if all of the $T_i$s succeed, then failure can only occur due to hash collision which occurs with probability at most $\delta$ since we are using $\Theta(\log (n/\delta))$ bit hash functions and hashing a total of $O(n \log n)$ strings ($O(n)$ per level, and $O(\log n)$ levels).

Alice's part of the protocol takes $O(n \lceil \log (n/k^3)\rceil) = O(n \log (n /k))$ time.  Each $T_i$ takes $O(k \log (n/\delta) \lceil \log_k \log n\rceil)$ space after replication, and $T^*$ takes $O(k(\log (n/\delta) + k^2)\lceil \log_k \log n\rceil)$ space, so the total communication cost is $O((k \log(n/k) \log (n/\delta)+k^3) \lceil \log_k \log n\rceil)$.

Computing the $y_{i,j'}$s takes Bob $O(n \log(n/k) \log_n(n/\delta))$ time.  Generating each $H_{i,j}$ takes Bob $O(k \log_n(n/\delta))$ time, so generating all of them takes $O(n \log_n(n/\delta)/k)$ time.  Bob can compute the $T_{i,j,m,\ell}$s and their replications in $O((n+ k^3  \log(n/k))\log_n(n/\delta)\lceil \log_k \log n\rceil)$ time.  To achieve this, he first computes all of the $T_{i,j,m,i}$s (note that here $\ell = i$).  Each one takes $O(k \log_n(n/\delta))$ time since it takes $O(k \log_n(n/\delta))$ time to initialize and then $O(\log_n(n/\delta))$ time to insert the single item into the table.  Now we observe that $T_{i,j,m,\ell}$ for $\ell > i$ is equal to the sum of $T_{i+1,2j-1,m,\ell}$ and $T_{i+1,2j,m,\ell}$, thus once we have computed each $T_{i,j,m,i+t}$, we can compute a given $T_{i,j,m,i+t+1}$ in $O(k)$ time by adding together two already computed IBLTs.  Thus the total time to compute the $T_{i,j,m,\ell}$s is $O(k \log_n(n/\delta))$ times how many of them there are, giving us a total precomputation time of
\begin{align*}O&(n \log(n/k) \log_n(n/\delta) + \lceil \log_k \log n\rceil\sum_{i=1}^{\lceil \log_2(n/k^3) \rceil} \sum_{j=1}^{2^i k} \sum_{m=-k}^k \sum_{\ell=i}^{\lceil \log_2(n/k^3) \rceil} O(k \log_n(n/\delta)) \\
&= O\left(n  \log(n/k) \log_n(n/\delta)  + \lceil \log_k \log n\rceil\sum_{i=1}^{\lceil \log_2(n/k^3) \rceil} (\lceil \log_2(n/k^3) \rceil-i) 2^i k^3 \log_n(n/\delta) \right) \\
&= O\left(n  \log(n/k) \log_n(n/\delta) + k^3 \log(n/k) \log_n(n/\delta) \lceil \log_k \log n\rceil\right).
\end{align*}

Finally we examine Bob's computation time after he receives Alice's message.  $|L| = O(k)$ at all times, so the loop over $\ell$ takes $O(k^2 \log(n/k) \log_n(n/\delta) \lceil \log_k \log n\rceil)$ time, which is the entire computation in the failure check phase.  Bob's remaining two steps take $O(n \lceil \log_k \log n\rceil)$ time.
\end{proof}

Now with this document exchange protocol in hand, we get the following directory reconciliation result.

\thmfastsos*

\begin{proof}
Basically, we plug \autoref{thm:fastims} into \autoref{lem:ibltofdocs}.  We do all of Bob's precomputation steps once at the beginning of the protocol, and when trying to decode a document encoding if it fails by the end of the failure check step, we just stop there so for an appropriate choice of $\delta$, we should only have perform an $\tilde{O}(n)$ computation once for each of Bob's documents and once for each of Alice's documents that we recover.

We choose $\delta = 1/\poly(n)$ so that with at least $14/15$, in none of \autoref{lem:ibltofdocs}'s $O(d^2 \log d)$ attempts to decode a message, will the decoding fail after passing the failure check step.  Our bounds are then \autoref{lem:ibltofdocs}'s with $g(n',k)=O((k \log n' \log n +k^3) \lceil \log_k \log n'\rceil)$ and $f(n',k)=O(k^2 \log (n'/k) \log_{n'} n\lceil \log_k \log n'\rceil)$ with an additional $O(n \log h \log \min(d,h))$ computation time for Alice's time to generate her document exchange messages, \[O(n  \log h \log \min(d,h) + d^3 \log h  \log \min(d,h) \log \log h)\] time for Bob to perform his precomputations, and $O(d^2 \log h  \log \log h)$ time for Bob to perform his recovery of Alice's documents that have passed their failure checks. Our communication bound is then
\begin{align*}
O&\left(d\log s +  d \log \log \min(d,h) + \sum_{i=1}^{\log \min(d,h)} (d/2^i)g(2^i, h)\right) \\
&= O\left(d\log s +  d \log \log \min(d,h) + \sum_{i=1}^{\log \min(d,h)} d (\log h \log n + 2^{2i})  \log \log h \right) \\
&= O(d\log s+ d \log n \log h \log \min(d,h) \log \log h + d \min(d,h)^2 \log \log h).
\end{align*}

We have already factored in the construction time for Alice and Bob's messages in \autoref{lem:ibltofdocs}, so the \[O\left(\sum_{i=1}^{\log \min(d,h)} \sum_{j=1}^m f(2^i, n_j)\right)\] piece of the computation time is not included.

Thus our computation time from \autoref{lem:ibltofdocs} is simply
\begin{align*}
O&\left(\sum_{i=1}^{\log \min(d,h)} \sum_{j=1}^{\min(\q, d/2^i)} (\q + \min(\q, d/2^i) - 2j) f(2^i,n_j)\right) \\
&= O\left(\sum_{i=1}^{\log (d/\q)} \sum_{j=1}^{d/2^i} (\q + d/2^i) 2^{2i} \log n \log \log n_j +\sum_{i=\log (d/\q)}^{\log d} \sum_{j=1}^{\q} 2\q 2^{2i} \log n \log \log n_j  \right)\\
&= O\left(\sum_{i=1}^{\log (d/\q)} (\q + d/2^i) d 2^i \log n \log \log h +\sum_{i=\log (d/\q)}^{\log d} 2\q^2 2^{2i} \log n \log \log h\right)\\
&= O(d^2 \q^2 \log n \log \log h).
\end{align*}
Adding up all of our pieces, we get the desired computation time.
\end{proof}

\section{Details of Protocols for Unknown $d$} \label{app:unknownd}
Here we develop an efficient directory reconciliation protocol which uses only a constant number of rounds of communication for the case when we do not have a bound on $d$.  First we describe the tools we need for it.  The first is the CGK embedding of \cite{chakraborty2016streaming}. 
Recall that $\Delta_e$ is the edit distance function and let $\Delta_H$ be the Hamming distance function.

\begin{lemma}[Theorem 4.1 of \cite{chakraborty2016streaming}]
There is a mapping $f : \{0,1\}^n \times \{0,1\}^{6n}\rightarrow \{0,1\}^{3n}$ with the following properties:
\begin{enumerate}
\item The mapping can be computed in $O(n)$ time.
\item For every $x \in \{0,1\}^n$, given $f(x,r)$ and $r$, $x$ can be computed in $O(n)$ time with probability at least $1-\exp(-\Omega(n))$ over $r$.
\item For every $x, y \in \{0,1\}^n$, $\Delta_e(x,y)/2 \leq \Delta_H(f(x,r),f(y,r)) \leq 1300 (\Delta_e(x,y))^2$ with probability at least $2/3$ over $r$.
\end{enumerate}
\end{lemma}

\cite{chakraborty2016streaming} uses this embedding to produce a document exchange protocol by combining it with the following Hamming distance sketch.  Document exchange can then be performed by encoding each party's string using the same $r$, reconciling these encodings using the Hamming distance sketch, and then inverting the encoding.  

\begin{lemma}[Theorem 4.4 of \cite{porat2007improved}]
Given $x \in \{0,1\}^n$ and $k \in [n]$, there is an algorithm that produces an $O(k \log n)$ sketch such $s_k(x)$ in time $O(n \log n)$.  Given $s_k(x)$ and $s_k(y)$ for $y \in \{0,1\}^n$ and $\Delta_H(x,y) \leq k$, there is an algorithm taking time $O(k \log n)$ which returns all tuples $(x_i,y_i)$ for which $x_i \neq y_i$ with probability at least $1-1/n$ (over the random bits in the sketching algorithm).
\end{lemma}

Using the Hamming sketch requires knowing an upperbound $k$ on the Hamming distance between the strings.  we estimate this difference efficiently using set difference estimators, which if we interpret sets as binary vectors, can be used to estimate Hamming distance.  In the language of Hamming distance, a set difference estimator is a data structure for estimating the Hamming distance between two binary strings.  It implicitly maintains two strings $x,y\in\{0,1\}^n$ and supports two operations: creation, merge and query.  Creation takes in a single string $x$ and makes an estimator $D$ representing $x$ and $y=\{0\}^n$.  Merge takes in a second set difference estimator $D'$, which implicitly maintains sets $x',y'\in\{0,1\}^n$ and returns a new set difference estimator $D''$ representing $x \wedge x'$ and $y \wedge y'$, where $\wedge$ denotes the logical OR operation.  Query returns an estimate for $\Delta_H(x,y)$.

\begin{lemma}[Theorem 3.1 of \cite{mitzenmacher2017reconciling}] \label{thm:strata}
There is a set difference estimator requiring $O(\log(1/\delta) \log n)$ space with $O(n \log(1/\delta)$ creation times, and $O(\log(1/\delta))$ merge and query times, which reports the size of the Hamming distance to within a constant factor with probability at least $1 - \delta$.
\end{lemma}

We now combine these tools to develop a directory reconciliation protocol for unknown $d$, designed based on the multi-round set of sets reconciliation protocol of \cite{mitzenmacher2017reconciling}.

\begin{lemma}
Directory reconciliation with unknown $d$ can be solved in 4 rounds using 
\[O(\q \log s \lceil \log_{\q} (1/\delta)\rceil+\q \log h \log^2(\q/\delta) + d^2 \log h \lceil \log_h (\q/\delta)\rceil)\]
bits of communication and 
\[O((sh + \q^2)\log^2(\q/\delta)+(\q h + d^2)\log h\lceil\log_h(\q / \delta)\rceil)\]
time with probability at least $1-\delta$.
\end{lemma}

\begin{proof}
First we detail the reconciliation protocol we use (except for some small amount of probability amplification), and then we argue its correctness afterwards.
\begin{enumerate}
\item Bob computes a $\Theta(\log (s/\delta))$-bit pairwise independent hash of his documents, creates a set difference estimator (with failure probability $\Theta(\delta)$) for his set of hashes, and sends it to Alice.
\item Alice computes a $\Theta(\log (s/\delta))$-bit pairwise independent hash of her documents.  She uses Bob's set difference estimator to estimate the size of the difference between their sets of hashes, which should be $O(\q)$.  Alice then inserts all of her document hashes into $O(\q)$-cell IBLT $T_A$ which she transmits to Bob.
\item Bob inserts all of his document hashes into an $O(\q)$-cell IBLT $T_B$.  Bob decodes $(T_A,T_B)$, and determines which of his child sets differ from Alice.  For each of his differing documents, he creates $\Theta(\log(\q/\delta))$ length $\Theta(h)$ CGK encodings of it, and constructs a set difference estimator (with failure probability $\Theta(\delta/\poly(\q))$) for each of the encodings.  For each of these documents, Bob creates a vector of its corresponding set difference estimators and inserts the vector into a list $L_B$.  He transmits $T_B$ and $L_B$ to Alice.
\item Alice decodes $(T_A,T_B)$, and constructs $L_A$, a list of vectors set difference estimators (each estimator within a vector again corresponding to a different CGK encoding) for each of her differing documents.  For each vector of set difference estimators $L_{A,i}\in L_A$ and each $L_{B,j} \in L_B$, Alice estimates the edit distance between the documents corresponding to $i$ and $j$ by merging $L_{A,i}$ and $L_{B,j}$ element-wise, and then taking the median of the estimates.  Let $b_i$ be the index $j$ of the $L_{B,j}$ with which $L_{A,i}$ yielded the smallest estimate, let $c_i$ be the index of the CGK embedding used in that estimate, and let $d_i$ be that estimated difference.  For each $i$, Alice transmits $b_i$, $c_i$, $d_i$, and $S_i$, a Hamming distance sketch (with $k=\Theta(d_i)$, then replicated $\Theta(\lceil\log_h(\q/\delta) \rceil)$ times) of the $c_i$th CGK encoding of Alice's document $i$.
\item For each of the received tuples $(b_i, c_i, d_i, S_i)$ pairs, Bob recovers Alice's document $i$ by creating a Hamming sketch (with $k = \Theta(d_i)$ and $\Theta(\lceil\log_h(\q/\delta) \rceil)$ replication) of the $c_i$th CGK encoding of his document $b_i$ and uses it to decode $S_i$.  Bob then applies the extracted differences to his CGK encoded document and inverts the CGK encoding to yield Alice's document $i$.  Bob then recovers Alice's total directory by removing all documents corresponding to $L_B$ from his set and adding in Alice's documents that he has recovered. 
\end{enumerate}

 This protocol succeeds so long as none of the hashes collide, none of the set difference estimators fail, $T_A$ and $T_B$ together decode, none of the Hamming sketches fail, and for each pair of differing documents, the median Hamming distance between their CGK encodings is accurate (up to $O(k^2)$).  Union bounding over all $O(s^2)$ pairs of documents, none of the hashes collide with probability at least $1 - O(\delta)$.  The first set difference estimator succeeds with probability $1-O(\delta)$.  $(T_A,T_B)$ decodes with probability at least $1 - \poly(\q)$.  By replicating step 2 (and the corresponding part of step 3) $\Theta(\lceil\log_{\q}(1/\delta)\rceil)$ times, we reduce the probability that $T_A$ and $T_B$ fails to decode to $O(\delta)$.

There are $O(\q^2 \log(\q/\delta))$ pairs of set difference estimators, each of which fails with probability $O(\delta / \poly(\q))$, so that they all succeed with probability at least $1 - O(\delta)$.  For each of the $O(\q^2)$ pairs of documents compared, there are $O(\log(\q/\delta))$ CGK encodings, each of which fails with probability at most $1/3$.  For each pair of documents, by a Chernoff bound, the CGK encoding with the median Hamming distance fails with probability $O(\delta/\poly(\q))$, thus by a union bound each pair's median CGK encoding succeeds with probability $1-O(\delta)$. Each Hamming sketch fails with probability $O(1/h)$ before replication, so by replicating it $\Theta(\lceil \log_h (\q/\delta)\rceil)$ times every pair of Hamming sketches will succeed with probability $1-O(\delta)$.  Putting it all together, for the right choice of constants, the protocol succeeds with probability at least $1 - \delta$.

Computing the hashes takes $O(n)$ time, and creating and transmitting the initial set difference estimator takes $O(\log(1/\delta))$ time and $O(\log(1/\delta)\log s)$ communication, by \autoref{thm:strata}.  Constructing and decoding $T_A$ and $T_B$ takes, over $O(\lceil\log_{\q}(1/\delta)\rceil)$ replications, $O(\lceil\log_{\q}(1/\delta)\rceil s)$ time and $O(\lceil\log_{\q}(1/\delta)\rceil\q\log s)$ bits of communication.  Computing all of the CGK encodings, and then later decoding them, takes $O(sh \log(\q/\delta))$ time.  By \autoref{thm:strata}, constructing $L_A$ and $L_B$ takes $O(sh\log^2 (\q / \delta))$ time and transmitting $L_B$ takes $O(\q \log h \log^2 (\q / \delta))$ bits of communication.  Finding the $b_i$s, $c_i$s, and $d_i$s consists of $O(\q^2\log (\q / \delta))$ set difference merges and queries, which by \autoref{thm:strata} take a total of $O(\q^2 \log^2 (\q / \delta))$ time.  Sending the $b_is$, $c_i$s, and $d_i$s takes $O(\q\log (\q h \log(1/\delta)))$ bits of communication.  Computing and decoding the $S_i$s takes $O((\q h + d^2) \log h \lceil\log_h(\q/\delta)\rceil)$ time and transmitting them takes $O(d^2 \log h \lceil\log_h(\q/\delta)\rceil)$ bits of communication.  Adding up all of these terms, we get our desired bounds.
\end{proof}

Plugging $\delta = 1/\poly \max(\q, h)$ into this lemma, then replicating the result $\Theta(\log n / \log (1/\delta))$ times in parallel, we get our theorem.

\multiroundnod*

\end{document}